\documentclass[final,leqno]{article}
\usepackage[T1]{fontenc} 
\usepackage{ae,aecompl}	 
\usepackage{times}
\usepackage{soul,fancyhdr}
\usepackage{amsfonts}\usepackage{amsmath}\usepackage{amssymb}
\usepackage[hang,small,bf]{caption}
\usepackage[latin1]{inputenc} 
\usepackage{hyperref}\usepackage{breakurl} 
\usepackage{graphicx} 
\usepackage{float,booktabs}
\usepackage{subfig}
\usepackage[usenames,dvipsnames]{color}

\newtheorem{theorem}{Theorem}[section]
\newtheorem{example}[theorem]{Example}
\newtheorem{remark}[theorem]{Remark}
\newtheorem{lemma}[theorem]{Lemma}
\newenvironment{proof}[1][Proof]{\textbf{#1.} }{\ \rule{0.5em}{0.5em}}

\newcommand{\comment}[1]{}     

\begin{document}

\title{Dengue model with early-life stage of vectors and age-structure within host}

\maketitle

\begin{center}
\author{Fabio Sanchez\footnote{Centro de Investigaci\'on en Matem\'atica Pura y Aplicada (CIMPA), Escuela de Matem\'atica, Universidad de Costa Rica. San Pedro de Montes de Oca, San Jos\'e, Costa Rica, 11501. Email: \url{fabio.sanchez@ucr.ac.cr}} and Juan G. Calvo\footnote{Centro de Investigaci\'on en Matem\'atica Pura y Aplicada (CIMPA), Escuela de Matem\'atica, Universidad de Costa Rica. San Pedro de Montes de Oca, San Jos\'e, Costa Rica, 11501Email: \url{juan.calvo@ucr.ac.cr}}}
\end{center}

\begin{abstract}
We construct an epidemic model for the transmission of dengue fever with an early-life stage in the vector dynamics and age-structure within hosts. The early-life stage of the vector is modeled via a general function that supports multiple vector densities. The {\it basic reproductive number} and {\it vector demographic threshold} are computed to study the local and global stability of the infection-free state. A numerical framework is implemented and simulations are performed.
\end{abstract}

\section{Introduction} \label{sec:intro}
Dengue fever has been a burden to public health officials in the tropics and subtropics since the 20th century \cite{cdc,harris2000}. Dengue virus, of the genus \textit{Flavivirus} of the family Flaviviridae, is an infectious disease transmitted by the mosquitoes {\it Aedes aegypti} and {\it Aedes albopictus} \cite{gubler1998}. There are four serotypes of the dengue virus, called DEN-1, DEN-2, DEN-3, and DEN-4. After infection of one serotype, the infected person acquires lifelong immunity for that specific serotype and short-term immunity to other serotypes \cite{cdc}. 

There are two stages in the transmission cycle of dengue that have been reported. There is an enzootic transmission cycle between primates mostly in forests, with transmission between vectors through feeding from infected animals \cite{gubler1998}. These infected mosquitoes rarely wander far from the forest, so the infection to human populations comes from humans or livestock who visit a forest with presence of the virus, encounter an infected vector, become infected, and then infect the mosquitoes in their population center, which then can spread the disease to the rest of the population. In the case of rural, small populations, since the population rapidly gets saturated with the infection and subsequently immunized, the epidemic usually is short-termed. The other stage is between vectors and humans, where vectors bite an infected human and can potentially become infected. In this work we will focus on the interaction between vectors and humans.


The model in \cite{sanchez2006} focuses on the early-life stage of the vector and explores the effect of multiple vector densities on dengue outbreaks. Previous work on dengue models mostly focus on the adult vector-host interactions; see, e.g., \cite{feng1997,esteva1998,esteva1999,sanchez2012,manore2014,murillo2014,brauer2016,sanchez2018}.

The model we present here incorporates age-structure within the host population, as well as the early-life stage of the vector as in \cite{sanchez2006}. The inclusion of age-structure in the human/host population can help to determine prevention and control strategies based on specific population age groups and other social factors inherent to a subgroup of the population at risk. 

This article is organized as follows. In Section \ref{sec:model}, we outline the mathematical model. In Section \ref{sec:R0}, we compute the {\it vector demographic threshold}, the {\it basic reproductive number}, and determine the stability of the system. Section \ref{sec:num} includes the numerical scheme and numerical simulations, confirming the theoretical results. Finally, in Section \ref{sec:disc} we present some relevant conclusions and final thoughts.

\section{Mathematical model} \label{sec:model}
We consider a compartmental model with age-structure within the host population, with susceptible, infectious and recovered hosts, denoted by $S_h(t,a)$, $I_h(t,a)$ and $R_h(t,a)$, respectively.
Vectors are described by three states: $E(t)$ (egg/larvae at time $t$), $S_v(t)$ (number of non-infected vectors) and $I_v(t)$ (number of infected vectors). 

Hosts and vectors are coupled via a transmission process, where susceptible hosts can become infected at rate $\beta(a) \frac{I_v(t)}{N_v(t)}$, where $\beta(a)$ represents the age-dependent contact rate (vector-human) and $N_v(t)=S_v(t)+I_v(t)$ is the total number of vectors in the system. The number of new hosts coming into the system, $\Lambda$, is assumed to be constant. Infected individuals can recover at rate $\gamma(a)$,
and all hosts exit the system at rate $\mu_h(a)$.


We will restrict ourselves to the case of proportional mixing: $$ p(t,a) =  \frac{c(a)n(t,a)}{\int_0^\infty c(a) n(t,a)\ da}, $$ with $c(a)$ the age-specific per-capita contact/activity rate. We then define the force of infection $$B(t) = \int_0^\infty \frac{I_h(t,a)}{n(t,a)} p(t,a)da,$$ where $n(t,a)= S_h(t,a)+I_h(t,a)+R_h(t,a)$ is the population density and $\int_{0}^{\infty}n(t,a)da$ is the total population. 

In the vector classes, we consider that eggs enter the system at rate $f(N_v)$, where $f(N_v)$ is assumed to be a Kolmogorov type function, $f(N_v)=N_v g(N_v)$, with $g:\mathbb{R}^+\rightarrow \mathbb{R}^+$ a differentiable function such that $g(0)>0$, $g(\infty)=0$. They also exit the system at rate $\mu_e$ and become mosquitoes at rate $\delta$. Mosquitoes become infectious at rate $\beta_v B(t)$, where $\beta_v$ is the transmission rate from humans to mosquitoes and $B(t)$ is the force of infection. Mosquitoes also exit the system at rate $\mu_v$. In our analysis, we assume that parameters $\delta, \mu_e, \mu_v, \beta_v$ are constant and $\beta, \mu_h, \gamma$ are continuous functions on age.

The model we just described is given by the system of differential equations:
\begin{subequations} \label{eq:system}
\begin{align} 
\begin{split}
\frac{dE}{dt} &= f(N_v)-(\delta+\mu_e)E,
\end{split}\\
\begin{split}
\frac{dS_v}{dt} &= \delta E-\beta_v S_v B(t) -\mu_v S_v,
\end{split}\\
\begin{split}
\frac{dI_v}{dt} &= \beta_v S_v B(t)-\mu_v I_v,
\end{split}\\
\begin{split}
\Big(\frac{\partial}{\partial t}+\frac{\partial}{\partial a}\Big)S_h(t,a) &= -\beta(a)S_h(t,a)\frac{I_v}{N_v}-\mu_h(a)S_h(t,a),
\end{split}\\
\begin{split}
\Big(\frac{\partial}{\partial t}+\frac{\partial}{\partial a}\Big)I_h(t,a) &= \beta(a)S_h(t,a)\frac{I_v}{N_v}-(\mu_h(a)+\gamma(a))I_h(t,a),
\end{split}\\
\begin{split}
\Big(\frac{\partial}{\partial t}+\frac{\partial}{\partial a}\Big)R_h(t,a) &= \gamma(a)I_h(t,a)-\mu_h(a)R_h(t,a),
\end{split}
\end{align}
along with initial conditions given by
\begin{equation*}
\begin{array}{rlrlrl}
E(0) &= E_0, & S_v(0) &=S_{v_0}, & I_v(0) &=  I_{v_0} \\
S_h(t,0) &= \Lambda, & I_v(t,0) &=0, & R_h(t,0)&=0,\\
S_h(0,a) &= S_{h_0}(a), & I_h(0,a) &= I_{h_0}(a), & R_h(0,a)&=R_{h_0}(a).
\end{array}    
\end{equation*}
\end{subequations}

The total host population satisfies $$\Big(\frac{\partial}{\partial t}+\frac{\partial}{\partial a}\Big)n(t,a) = -\mu_h(a) n(t,a),$$
and we then can compute explicitly that 
\begin{equation} \label{eq_nta}
n(t,a) = 
\left\lbrace
\begin{array}{cl}
n_0(a-t) \dfrac{\mathcal{F}(a)}{\mathcal{F}(a-t)} & {\rm if }\ a\geq t,  \\
\Lambda \mathcal{F}(a) & {\rm if }\ a<t,  \\
\end{array}
\right.
\end{equation}
where $$\mathcal{F}(a) = e^{-\int_0^a \mu_h(s)\ ds}$$
is the proportion of individuals that survive at age $a$. Therefore, we define
\begin{align*}
    n^*(a) &:= \lim_{t\rightarrow \infty} n(t,a) = \Lambda \mathcal{F}(a),\\
    p_\infty(a) &:=  \lim_{t\rightarrow \infty} p(t,a) = \frac{c(a)\mathcal{F}(a)}{\int_0^\infty c(a) \mathcal{F}(a)\ da}.
\end{align*}

\section{Model analysis}\label{sec:R0}
In this section we explore the conditions for multiple vector demographic steady states and determine their stability. We also compute  the {\it basic reproductive number} and analyze local and global stability for the solutions of System \eqref{eq:system}.

\subsection{Vector demographic number, $\mathcal{R}_v$} 
Since $n(t,a)$ is given explicitly in \eqref{eq_nta}, we first rescale variables 
$$s_h(t,a) = \frac{S_h(t,a)}{n(t,a)},\quad i_h(t,a) = \frac{I_h(t,a)}{n(t,a)},\quad r_h(t,a) = \frac{R_h(t,a)}{n(t,a)},$$ 
to obtain the equivalent system
\begin{subequations} \label{age_vsys_2}
\begin{align}
\begin{split}
\frac{dE}{dt} &= f(N_v)-(\delta+\mu_e)E ,
\end{split}\\
\begin{split}
\frac{dS_v}{dt} &= \delta E-\beta_v  B(t) S_v -\mu_v S_v,
\end{split}\\
\begin{split} \label{eq_Iv}
\frac{dI_v}{dt} &= \beta_v B(t) S_v-\mu_v I_v,
\end{split}\\
\begin{split}
N_v &= S_v+I_v,
\end{split}\\
\begin{split}
\Big(\frac{\partial}{\partial t}+\frac{\partial}{\partial a}\Big)s_h(t,a) &= -\beta(a)s_h(t,a)\frac{I_v}{N_v},
\end{split}\\
\begin{split} \label{eq_ih}
\Big(\frac{\partial}{\partial t}+\frac{\partial}{\partial a}\Big)i_h(t,a) &= \beta(a)s_h(t,a)\frac{I_v}{N_v}-\gamma(a)i_h(t,a),
\end{split}\\
\begin{split}
\Big(\frac{\partial}{\partial t}+\frac{\partial}{\partial a}\Big)r_h(t,a) &= \gamma(a)i_h(t,a),
\end{split}\\
\begin{split}
B(t) &= \int_0^{\infty}p(t,a)i_h(t,a)da.
\end{split}
\end{align}
\end{subequations}
For a given equilibrium state $(E^*, S_{v}^{*}, I_{v}^{*},N_{v}^{*}, s_{h}^{*}(a), i_{h}^{*}(a), r_{h}^{*}(a),B^*)$ of System \eqref{age_vsys_2}, we study its local stability by using the perturbations
\begin{eqnarray*} 
E(t) &=& E^*+e^{\psi t}\widetilde{E},\\
S_v(t) &=& S_{v}^{*}+e^{\psi t}\widetilde{S_v},\\
I_v(t) &=& I_{v}^{*}+e^{\psi t}\widetilde{I_v},\\%
N_v(t) &=& N_{v}^{*}+e^{\psi t}\widetilde{N_v},\\
s_h(t,a) &=& s_{h}^{*}(a)+e^{\psi t}\widetilde{s_h}(a),\\
i_h(t,a) &=& i_{h}^{*}(a)+e^{\psi t}\widetilde{i_h}(a),\\
r_h(t,a) &=& r_{h}^{*}(a)+e^{\psi t}\widetilde{r_h}(a),\\
B(t) &=& B^* +  e^{\psi t}\widetilde{B},
\end{eqnarray*}
where $$B^* = \int_0^\infty p_\infty(a) i_{h}^{*}(a)\ da, \quad \widetilde{B} =  \int_0^\infty p_\infty(a) \widetilde{i_h}(a)\ da.$$

Linearization leads to the eigenvalue problem
\begin{subequations}
\begin{align}
    \psi \widetilde{E} &= f^\prime(N_{v}^{*})\widetilde{N_v}-(\mu_e+\delta)\widetilde{E} \label{eqE},\\
    \psi \widetilde{N_v} &= \delta \widetilde{E} - \mu_v \widetilde{N_v} \label{eqL},\\
    \psi \widetilde{I_v} &= \beta_v (\widetilde{S_v}B^* + S_{v}^{*}  \widetilde{B} ) -  \mu_v \widetilde{I_v} \label{eqJ},\\
    \psi \widetilde{S_v} &=\delta \widetilde{E}-\beta_v ( B^* \widetilde{S_v}+ S_{v}^{*}  \widetilde{B})-\mu_v \widetilde{S_v},\\
     \dfrac{d}{da}\widetilde{s_h}(a) +  \psi \widetilde{s_h}(a) &= -\beta(a)\left( \frac{I_{v}^{*}}{N_{v^*}} \widetilde{s_h}(a) +  \frac{\widetilde{I_v}}{N_{v}^{*}}s_{h}^{*}(a) -\frac{I_{v}^{*}}{N_{v}^{*}} \frac{\widetilde{N_v}}{N_{v}^{*}} s_{h}^{*}(a)\right),\\
     \dfrac{d}{da}\widetilde{i_h}(a) +  \psi \widetilde{i_h}(a) &=\beta(a)\left( \frac{I_{v}^{*}}{N_{v}^{*}} \widetilde{s_h}(a) +  \frac{\widetilde{I_v}}{N_{v}^{*}}s_{h}^{*}(a) -\frac{I_{v}^{*}}{N_{v}^{*}} \frac{\widetilde{N_v}}{N_{v}^{*}} s_{h}^{*}(a)\right)-\gamma(a)\widetilde{i_h}(a) \label{eqI},\\
     \dfrac{d}{da}\widetilde{r_h}(a) +  \psi \widetilde{r_h}(a) &=\gamma(a)\widetilde{i_h}(a).
\end{align}
\end{subequations}
For $\widetilde{E}, \widetilde{N_v} \neq 0$, equations \eqref{eqE} and \eqref{eqL} imply that
\begin{equation*}
    \dfrac{\delta f'(N_{v}^*)}{(\mu_e+\delta+\psi)(\psi+\mu_v )} = 1.
\end{equation*}
Let $$\mathcal{H}_v(\psi) = \dfrac{\delta f'(N_{v}^*)}{(\mu_e+\delta+\psi)(\psi+\mu_v )}.$$
We then define the {\it demographic vector} number
\begin{equation*}
\mathcal{R}_v = {\cal H}_v(0) = \dfrac{f'(N_{v}^*)}{\phi},    
\end{equation*}
where $$\phi = \dfrac{(\delta+\mu_e) \mu_v}{\delta}$$
represents the proportion of eggs that survive to the adult stage. Recall that the rate that eggs enters the system is given by $f(N_v)=N_v g(N_v)$. We can then establish the following result:

\begin{lemma} \label{lem:stabNv}
Suppose that the set $g^{-1}(\phi) = \lbrace N_v\in (0,+\infty):g(N_v) = \phi \rbrace$ is non-empty. For each $N_v \in g^{-1}(\phi)$, there exists a positive vector state 
\begin{equation}\label{eq_eqState}
(E^*, N_v^*)=\left(\dfrac{\mu_v}{\delta} N_v, N_v \right),   
\end{equation}
which is locally asymptotically stable if $\mathcal{R}_v<1$ and unstable otherwise.
\end{lemma}
\begin{proof}
We have that $(E,N_v)$ satisfies the system
\begin{align} \label{eq:pde_E_Sv}
\frac{dE}{dt} &= f(N_v)-(\delta+\mu_e)E,\nonumber \\
\frac{dN_v}{dt} &= \delta E -\mu_v N_v,
\end{align}
with appropriate initial conditions. Therefore, fixed points satisfy
\begin{align*} 
f(N_v^*) &= (\delta+\mu_e)E^*,\nonumber \\
\delta E^* &=\mu_v N_v^*.
\end{align*}
Multiplying both equations and using the fact that $f(N_v^*)=N_v^* g(N_v^*)$, we then deduce that $g(N_v^*)= \phi$ (for $E^*N_v^*\neq 0$).  Thus, for each $N_v \in g^{-1}(\phi)$ there exists the positive state given in \eqref{eq_eqState}. 

Moreover, for a fixed state \eqref{eq_eqState}, the associated Jacobian to System \eqref{eq:pde_E_Sv} is given by
\begin{equation*}
    \left[ 
    \begin{array}{cc}
        -(\delta +\mu_e) & f'(N_v^*)  \\
        \delta & -\mu_v
    \end{array}
    \right],
\end{equation*}
which eigenvalues are given by $$\frac{1}{2} \left( -(\delta+\mu_e+\mu_v) \pm \sqrt{(\delta+\mu_e+\mu_v)^2-4\mu_v(\mu_e+\delta)\left(1-\mathcal{R}_v \right)} \right).$$
If $R_v<1$, 
we then conclude that both eigenvalues have negative real part and \eqref{eq_eqState} is locally stable.
\end{proof}
\begin{remark}{\rm \label{rem_gNv}
Since we assume that $f(N_v)=N_vg(N_v)$, it is straightforward to verify that 
$$\mathcal{R}_v = 1+g'(N_v^*) \dfrac{N_v^*}{\phi}.$$ 
Thus, equilibrium points given by \eqref{eq_eqState} are locally stable if $g'(N_v^*)<0$, and unstable otherwise; see Figure \ref{fig:stabNv*}.
}
\end{remark}

\begin{figure}[t!]
\centering
\includegraphics[width=.7\textwidth]{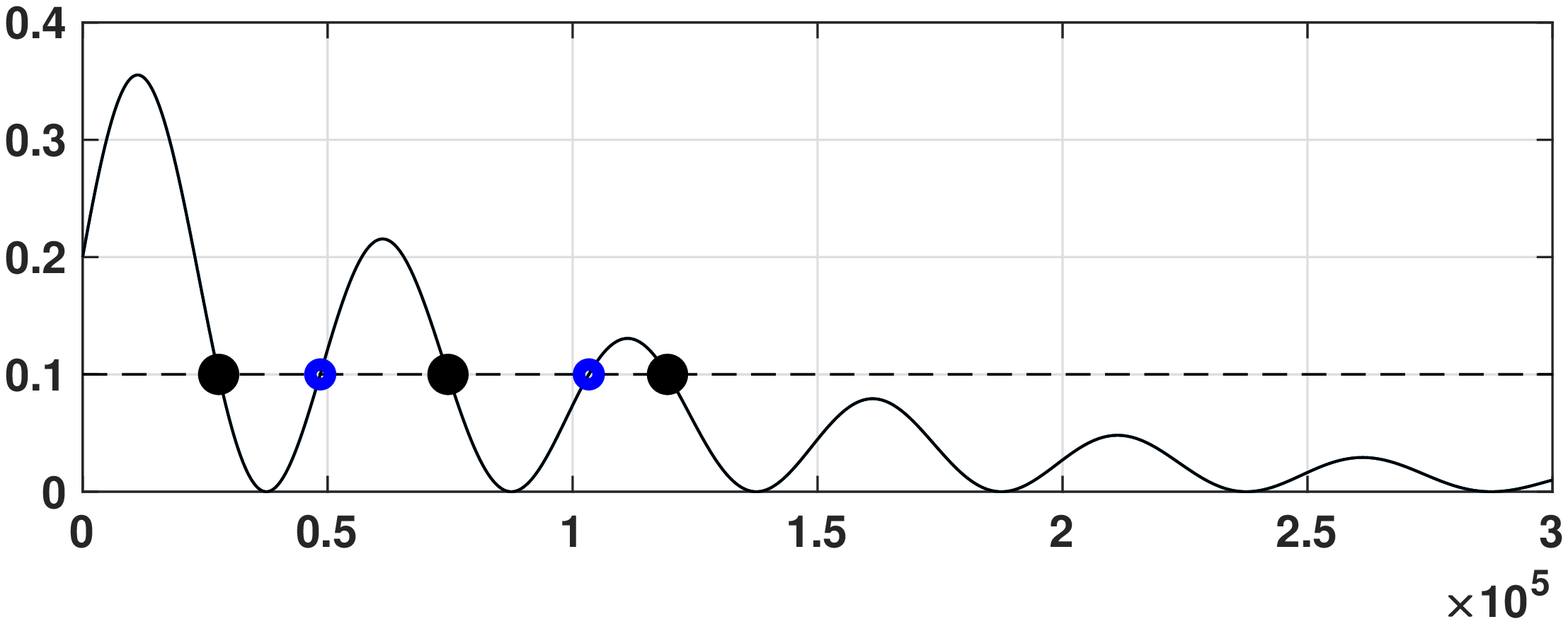}
\caption{An example for $g(N_v)$ as a function of $N_v$ for which multiple steady states exist. The dashed line corresponds to the value of $\phi$. Each intersection of both curves corresponds to an endemic state $ N_v^* \in g^{-1}(\phi)$. Black filled dots correspond to stable points since $g'(N_v^*)<0$ and circles correspond to unstable fixed points; see Lemma \ref{lem:stabNv} and Remark \ref{rem_gNv}. In this case, $N_v^*=0$ is unstable; see Lemma \ref{lem:Rv0}. \label{fig:stabNv*}}
\end{figure}

\begin{lemma} \label{lem:Rv0}
The vector-free state $(E^*,S_v^*,I_v^*) = (0,0,0)$ is locally asymptotically stable if $g(0)<\phi$, and unstable otherwise.
\end{lemma}
\begin{proof}
Suppose that $g(0)<\phi$. For $N_v^*=0$, $\mathcal{R}_v$ simplifies to $\mathcal{R}_v = \dfrac{g(0)}{\phi}<1$. Since ${\cal H}_v: [0,\infty)$ is a decreasing function of $\psi$, the equation $\mathcal{H}_v(\psi)=1$ can have only solutions with negative real part, and $(E^*,N_v^*)$ is locally stable. The result then follows since $S_v^*, I_v^*$ are non negative and $N_v^* = S_v^* + I_v^*$. If $g(0)>\phi$, then $\mathcal{H}_v(\psi)=1$ has one positive solution and the result holds.
\end{proof}

\subsection{Basic reproductive number, $\mathcal{R}_0$}
Consider now the disease-free state $$(E^*, S_{v}^{*}, I_{v}^{*},N_{v}^{*}, s_{h}^{*}(a), i_{h}^{*}(a), r_{h}^{*}(a),B^*) = (E^*,S_v^*,0,S_v^*,1,0,0,0)$$ for System \eqref{age_vsys_2}.
From \eqref{eqJ} we get 
$$\dfrac{\widetilde{I_v}}{N_{v}^*} = \dfrac{\beta_v \widetilde{B}}{\psi + \mu_v}.$$
Substituting in \eqref{eqI} and solving, we obtain that
\begin{equation*}
    \widetilde{{i}(a)} = \dfrac{\beta_v \widetilde{B}}{\psi + \mu_v} \int_0^a \beta(\tau) e^{-\int_\tau^a (\psi + \gamma(h))\ dh}\ d\tau.
\end{equation*}
Multiplying by $p_\infty(a)$ and integrating with respect to $a$, we deduce that
\begin{equation*}
    \widetilde{B} = \dfrac{\beta_v \widetilde{B}}{\psi + \mu_v} \int_0^\infty \int_0^a p_\infty(a) \beta(\tau) e^{-\int_\tau^a (\psi + \gamma(h))\ dh}\ d\tau\ da.
\end{equation*}
For $\widetilde{B}\neq 0$, we obtain that
\begin{equation*}
    \dfrac{\beta_v}{\psi+\mu_v} \int_0^\infty \int_0^a p_\infty(a) \beta(\tau) e^{-\int_\tau^a (\psi + \gamma(h))\ dh}\ d\tau\ da = 1.
\end{equation*}
Let
\begin{equation*}
    G(\psi) = \dfrac{\beta_v}{\psi+\mu_v} \int_0^\infty \int_0^a p_\infty(a) \beta(\tau) e^{-\int_\tau^a (\psi + \gamma(h))\ dh}\ d\tau\ da.
\end{equation*}
We then define the {\it basic reproductive number} 
\begin{equation} \label{eq_defR0}
    \mathcal{R}_0= G(0) = \dfrac{\beta_v}{\mu_v} \int_0^\infty \int_0^a p_\infty(a) \beta(\tau) e^{-\int_\tau^a \gamma(h)\ dh}\ d\tau\ da.
\end{equation}
In the particular case of constant parameters, it reduces to
\begin{equation*}
    \mathcal{R}_0 = \dfrac{\beta_v \beta}{\mu_v (\mu_h+\gamma)}.
\end{equation*}
We then have the following results:

\begin{theorem} \label{th:R0} 
Assume that $\mathcal{R}_0 < 1$. Then, the disease-free solution of System \eqref{age_vsys_2} is globally asymptotically stable.
\end{theorem}
\begin{proof}
From \eqref{eq_ih} we have that 
\begin{equation*}
    i_h(t,a)=\int_0^a \beta(\tau) e^{-\int_\tau^a \gamma(h)\ dh}  s_h(\tau+t-a,\tau) \dfrac{I_v(\tau+t-a)}{N_v(\tau+t-a}\ d\tau
\end{equation*}
for $t>a$. Multiplying by $p(t,a)$ and integrating with respect to $a$ we get
\begin{equation*}
        B(t) =\int_0^\infty \int_0^a \beta(\tau) p(t,a) e^{-\int_\tau^a \gamma(h)\ dh}  s_h(\tau+t-a,\tau) \dfrac{I_v(\tau+t-a)}{N_v(\tau+t-a)}\ d\tau.
\end{equation*}
Since $s(t,a)\leq 1$,
\begin{equation*}
        B(t) \leq \int_0^\infty \int_0^a \beta(\tau) p(t,a) e^{-\int_\tau^a \gamma(h)\ dh} \dfrac{I_v(\tau+t-a)}{N_v(\tau+t-a)}\ d\tau,
\end{equation*}
and therefore 
\begin{equation}\label{eq2}
        B^* \leq \dfrac{I_v^*}{N_v^*} \int_0^\infty \int_0^a \beta(\tau) p_\infty(a) e^{-\int_\tau^a \gamma(h)\ dh} \ d\tau.
\end{equation}
From \eqref{eq_Iv}, if $S_v^* = 0$ then $I_v^*=0$. Otherwise, it holds that
\begin{equation}\label{eq3}
    B^* = \dfrac{\mu_v}{\beta_v} \dfrac{I_v^*}{S_v^*}.
\end{equation}
Combining \eqref{eq2}, \eqref{eq3}, and using \eqref{eq_defR0}, we obtain that
\begin{equation*}
    0\leq \dfrac{I_v^*}{N_v^*} \leq \dfrac{I_v^*}{S_v^*}\leq \dfrac{I_v^*}{N_v^*} \mathcal{R}_0,
\end{equation*}
since $S_v^*\leq N_v^*$. By assumption, $\mathcal{R}_0<1$, and therefore $I_v^*=0$. Hence, $B^*=0$ and $i_h^*(a)=0$.
\end{proof}

\begin{theorem} \label{th:endemic}
Assume that there exists $N \in \lbrace N\in (0,+\infty):g(N) = \phi \rbrace$ with $g'(N)<0$. If $\mathcal{R}_0>1$, there exists one endemic non-uniform stable steady state for System \eqref{age_vsys_2}.
\end{theorem}
 
\begin{proof}
It is straightforward to verify that there exists the endemic state
\begin{align*}
    N_v^* &= N,\\
    E^* &= N \dfrac{\mu_v}{\delta}, \\
    S_v^* &= N \frac{\mu_v}{(\mu_v + \beta_v  B^*)},\\
    I_v^* &= N \frac{\beta_v  B^* }{(\mu_v + \beta_v  B^*)}.
\end{align*}
By hypothesis and Remark \ref{rem_gNv}, it holds that $(E^*,N_v^*)$ is a positive local stable fixed point. We will prove that  $B^*>0$ which implies that $I_v^*>0$.

A non-uniform steady state for System \eqref{age_vsys_2} is a solution of the nonlinear system 
\begin{align} \label{eq:systB*}
    \dfrac{I_{v}^{*}}{N_{v}^{*}} &= \frac{\beta_v  B^*}{\mu_v + \beta_v  B^*},\nonumber\\
    \dfrac{d s_{h}^{*}(a)}{da} &= -\beta(a)s_{h}^{*}(a)\dfrac{I_{v}^{*}}{N_{v}^{*}},\nonumber\\
    \dfrac{d i_{h}^{*}(a)}{da} &= \beta(a)s_{h}^{*}(a)\dfrac{I_{v}^{*}}{N_{v}^{*}}-\gamma(a) i_{h}^{*}(a),\\
    \dfrac{d r_{h}^{*}(a)}{da} &= \gamma(a) i_{h}^{*}(a),\nonumber\\
    B^* &= \int_0^\infty p_\infty(a) i_{h}^{*}(a)\ da,\nonumber
\end{align}
for $a>0$, with initial conditions $$s_{h}^{*}(0)=1,\ i_{h}^{*}(0)=0,\ r_{h}^{*}(0)=0.$$
Consider the linear system of equations with parameter $B$ given by
\begin{align}\label{eq:systB}
    \dfrac{d s^*_{h_B(a)}}{da} &= -\beta(a)s_{h}^{*}(a)\frac{\beta_v  B}{(\mu_v + \beta_v  B)},\nonumber\\
    \dfrac{d i^*_{h_B(a)}}{da} &= \beta(a)s_{h}^{*}(a)\frac{\beta_v  B}{(\mu_v + \beta_v  B)}-\gamma(a) i_{h}^{*}(a),\\
    \dfrac{d r^*_{h_B(a)}}{da} &= \gamma(a) i_{h}^{*}(a),\nonumber
\end{align}
for $a>0$, with initial conditions $$s_{h_B}^{*}(0)=1,\ i_{h_B}^{*}(0)=0,\ r_{h_B}^{*}(0)=0.$$
Given the solution $(s_B^*(a),i_B^*(a),r_B^*(a))$ of system \eqref{eq:systB}, define $$H(B)= \int_0^\infty i_B^*(a) p_\infty(a)\ da.$$
It holds that $(s_B^*(a),i_B^*(a),r_B^*(a))$ satisfies system \eqref{eq:systB*} if and only if $B$ is a fixed point of $H$; i.e., $H(B) = B$. Moreover, if $B=0$ then $i_B^*(a)=0$ and $H(0)=0$. Thus, in order to guarantee existence of at least one non-trivial solution to \eqref{eq:systB*}, it is just necessary to prove that $H(B)$ has a positive fixed point.

Solving for $s_{h_B}^{*}(a)$ and $i_{h_B}^{*}(a)$, it follows that
\begin{align*}
    s_{h_B}^{*}(a) &= e^{\frac{- \beta_v  B}{\mu_v + \beta_v  B} \int_0^a \beta(h)\ dh},\\
    i_{h_B}^{*}(a) &= \frac{\beta_v  B}{\mu_v + \beta_v  B}\int_0^a e^{-\int_\tau^a \gamma(h)\ dh} \beta(\tau)s_{h_B}^*(\tau)\ d\tau.
\end{align*}
Define $$G(B):=\dfrac{H(B)}{B}\quad {\rm for\ } B\neq 0.$$ The function
$G:[0,1]\rightarrow\mathbb{R}$ is continuous by defining $G(0) = H'(0)$. We have that
$$G(0) = \frac{\beta_v}{\mu_v} \int_0^\infty \int_0^a e^{-\int_\tau^a \gamma(h)\ dh} \beta(\tau) p_\infty(a)\ d\tau\ da = \mathcal{R}_0 > 1$$
and $G(1)=H(1)<1$. Therefore, there exists $B^* \in (0,1)$ such that $G(B^*)=1$, i.e., $H(B^*)=B^*$ and there exists at least one endemic state for $i_h^*(a)$. Moreover, $B^*$ is unique since $G(B)$ is strictly decreasing.
In particular, $B^*>0$ implies that $I_v^*>0$, reaching an endemic steady state, both for vectors and humans.
\end{proof}

\section{Numerical implementation} \label{sec:num}

For simplicity, we discretize the system of partial differential equations \eqref{age_vsys_2} with a first-order upwind finite difference scheme. We approximate the solution on the physical domain of interest given by the rectangle $\lbrace(t,a):t\in [0,T], a\in [0,A]\rbrace$. We divide the intervals $[0,T]$ and $[0,A]$ in $N_T$ and $N_A$ subintervals, respectively, and consider the grid given by the nodes
$$(t_j,a_k) = \left( j\Delta t, k\Delta a\right),$$
for $j\in\lbrace 0,1,\ldots,N_T\rbrace$, $k\in\lbrace 0,1,\ldots,N_A\rbrace$, where 
$$\Delta t := \frac{T}{N_T},\quad \Delta a := \frac{A}{N_A}$$ are the corresponding step sizes. For any function $x$ and a grid point $(t_j,a_k)$, we denote the approximation of $x(t_j,a_k)$ by $x_k^j$. If the function depends only on age or time, it is denoted simply by $x_k$ or $x^j$. We approximate the force of infection $B(t_j)$ by the composite trapezium rule
$$B^j = \Delta a \left( \sum_{k=1}^{N_A-1} p_j^k\ (i_h)_j^k\ da + \frac{1}{2}p_j^{N_A} (i_h)_j^{N_A} \right).$$
We then have the explicit scheme given by the equations
\begin{align*}
    (N_v)^{j} &= (S_v)^{j} + (I_v)^{j},\\ 
    E^{j+1} &=  E^{j}+\Delta t\left[f\left((N_v)^{j}\right)-(\delta + \mu_e)E^{j}\right],\\
    (S_v)^{j+1} &= E^{j}+\Delta t\left[\delta E^{j} - \beta_v B^j (S_v)^j - \mu_v (S_v)^j\right],\\
    (I_v)^{j+1} &= (I_v)^{j}+\Delta t\left[ \beta_v B^j (S_v)^j - \mu_v (I_v)^j\right],\\
    (s_h)^{j+1}_k &= (s_h)^{j}_k + \Delta t \left[- \beta_k (s_h)^j_k \dfrac{(I_v)^j}{(N_v)^j}  - \frac{(s_h)^j_k - (s_h)^{j}_{k-1}}{\Delta a}\right],\\
    (i_h)^{j+1}_k &= (i_h)^{j}_k +  \Delta t\left[ \beta_k (s_h)^j_k \dfrac{(I_v)^j}{(N_v)^j} - \gamma_k (i_h)^j_k -  \frac{(i_h)^j_k - (i_h)^{j}_{k-1}}{\Delta a} \right],\\
    (r_h)^{j+1}_k &= (r_h)^{j}_k+ \Delta t\left[ \gamma_k (i_h)^j_k -  \frac{(r_h)^j_k - (r_h)^{j}_{k-1}}{\Delta a} \right],\\
\end{align*}
for $1\leq k\leq N_A$ and $0\leq j\leq N_T-1$. Thus, given the initial conditions $E^0$, $(S_v)^0 $, $(I_v)^0$, $s_0^j$, $s_k^0$, $i_0^j$, $i_k^0$, $r_0^j$, $r_k^0$, we can compute the values of the unknowns on the grid points in successive time steps. 

We present different scenarios where we confirm the results from Theorems \ref{th:R0} and \ref{th:endemic}. For the age dependent parameters, we show the distributions used in Figure \ref{fig:ex1_gamma_c}. We consider three different functions $f(N_v) = N_v g(N_v)$ in the following sections: a logistic-type function, a case with multiple vector steady-states, and a seasonal example.

\begin{figure}[t!]
\centering
\subfloat{\includegraphics[width=0.3\textwidth]{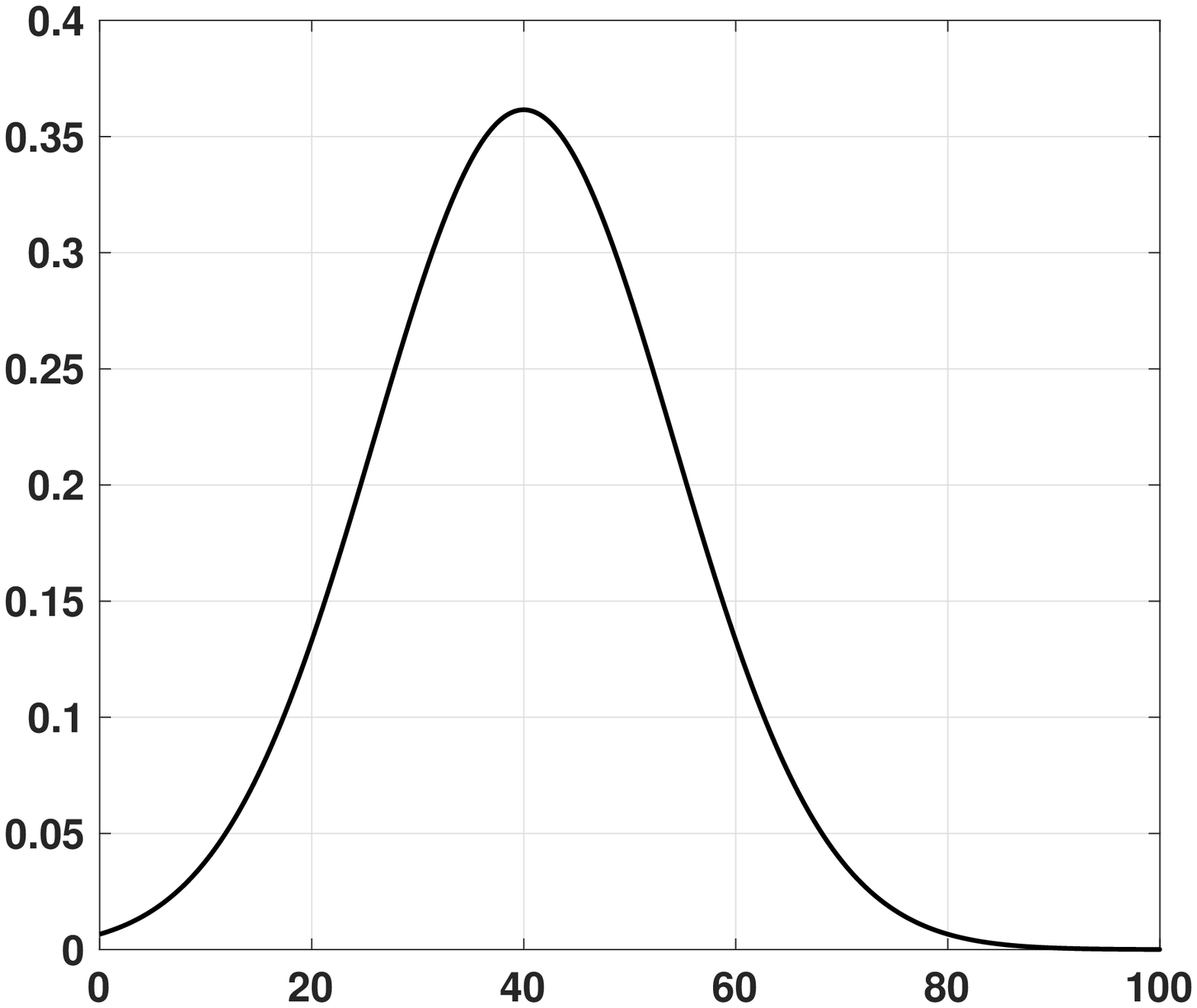}}\hfill
\subfloat{\includegraphics[width=0.3\textwidth]{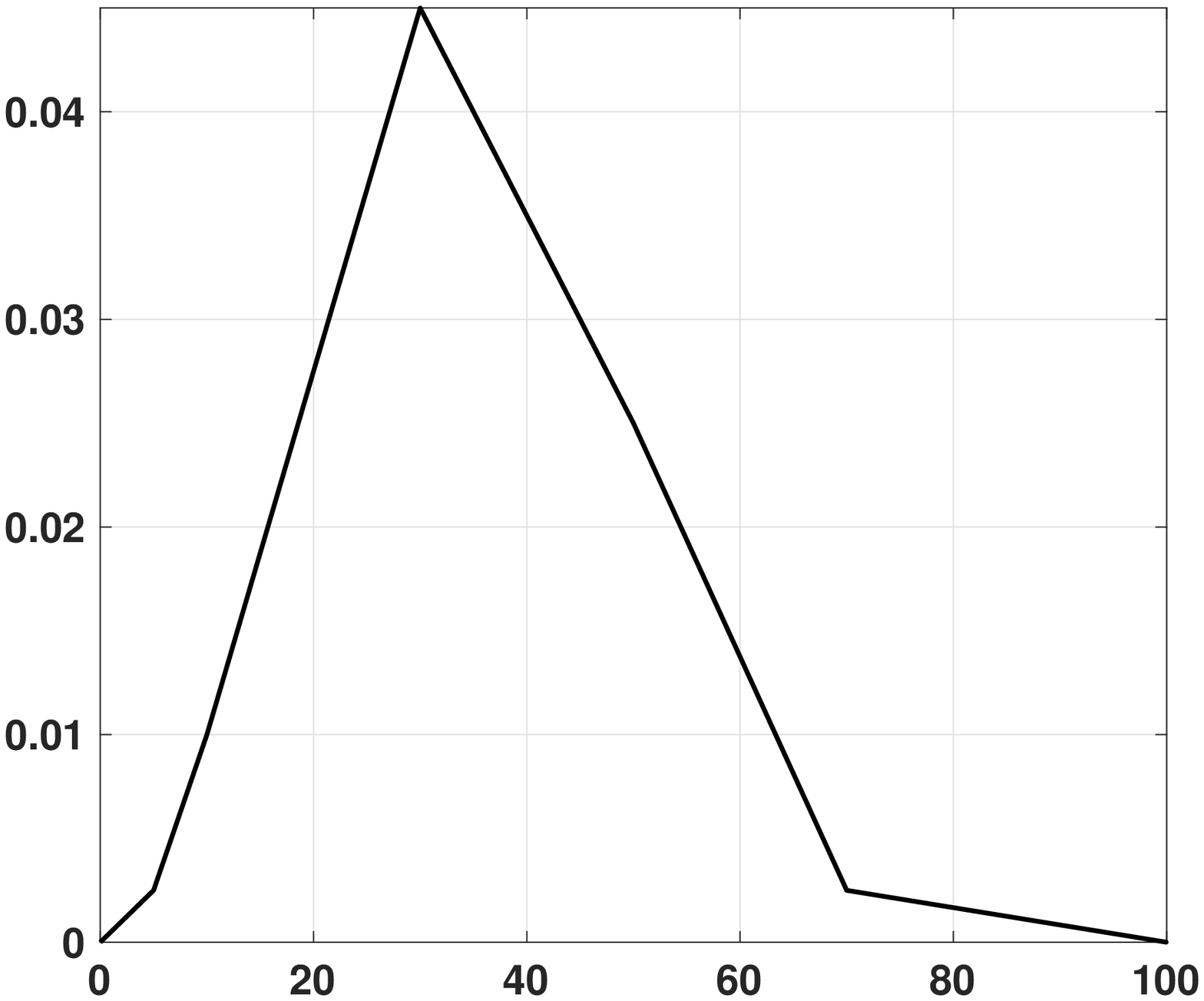}}\hfill
\subfloat{\includegraphics[width=0.3\textwidth]{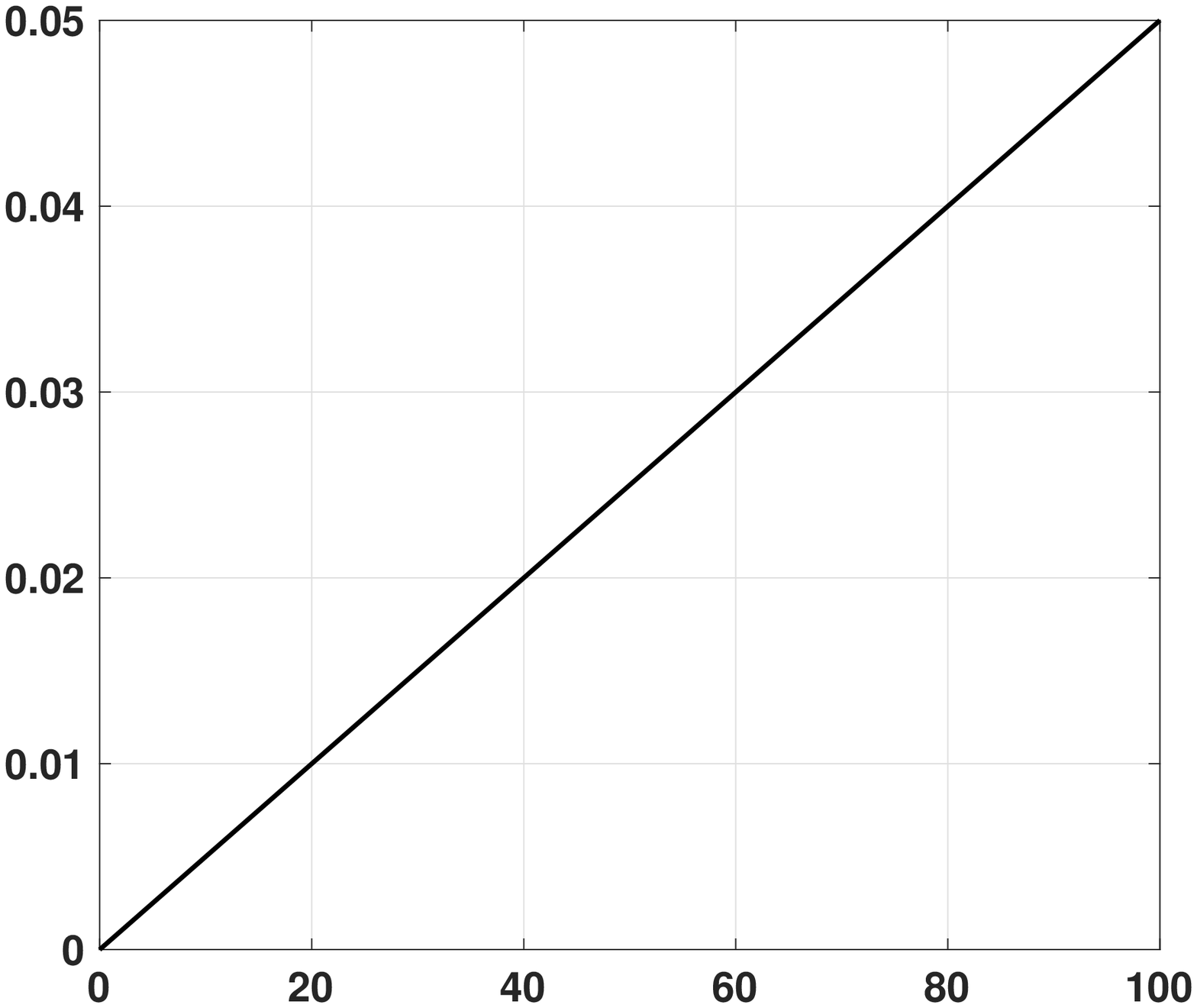}}
\caption{(left) Transmission rate $\beta(a)$, (middle) recovery rate $\gamma(a)$ and (right) mortality rate $\mu(a)$, as functions of age.
\label{fig:ex1_gamma_c}}
\end{figure}

\subsection{Logistic growth} \label{sec:numexp}
We first consider 
    \begin{equation*}
    g(N_v) = r\left (1-\dfrac{N_v}{N_{\max}}\right)    
    \end{equation*}
    for given constants $r$ (mosquito growth rate) and $N_{\max}$ (maximum number of mosquitoes that the system can hold); see Figure \ref{fig_g}. Recall that for a positive steady-state on vectors we require $\phi=g(N_v)$. In this case:
    \begin{enumerate}
        \item If $\phi>r$, there exists only the trivial state $N_v^*=0$, which is stable since $\mathcal{R}_v = \dfrac{g(0)}{\phi} <1$; see Figure \ref{fig_g2}.
        \item For $\phi < r$, besides the unstable state $N_v^*=0$, we have the non-zero state $$N_v^*= N_{\max}\left(1-\dfrac{\phi}{r}\right), $$
    which is stable since $g'(N_v^*)<0$; see Remark \ref{rem_gNv} and Figure \ref{fig_g1}. 
    \end{enumerate}

\begin{figure}[t!]
\centering
\subfloat[$N_v^*=0$ is the only (stable) fixed point.\label{fig_g2}]{\includegraphics[width=0.49\textwidth]{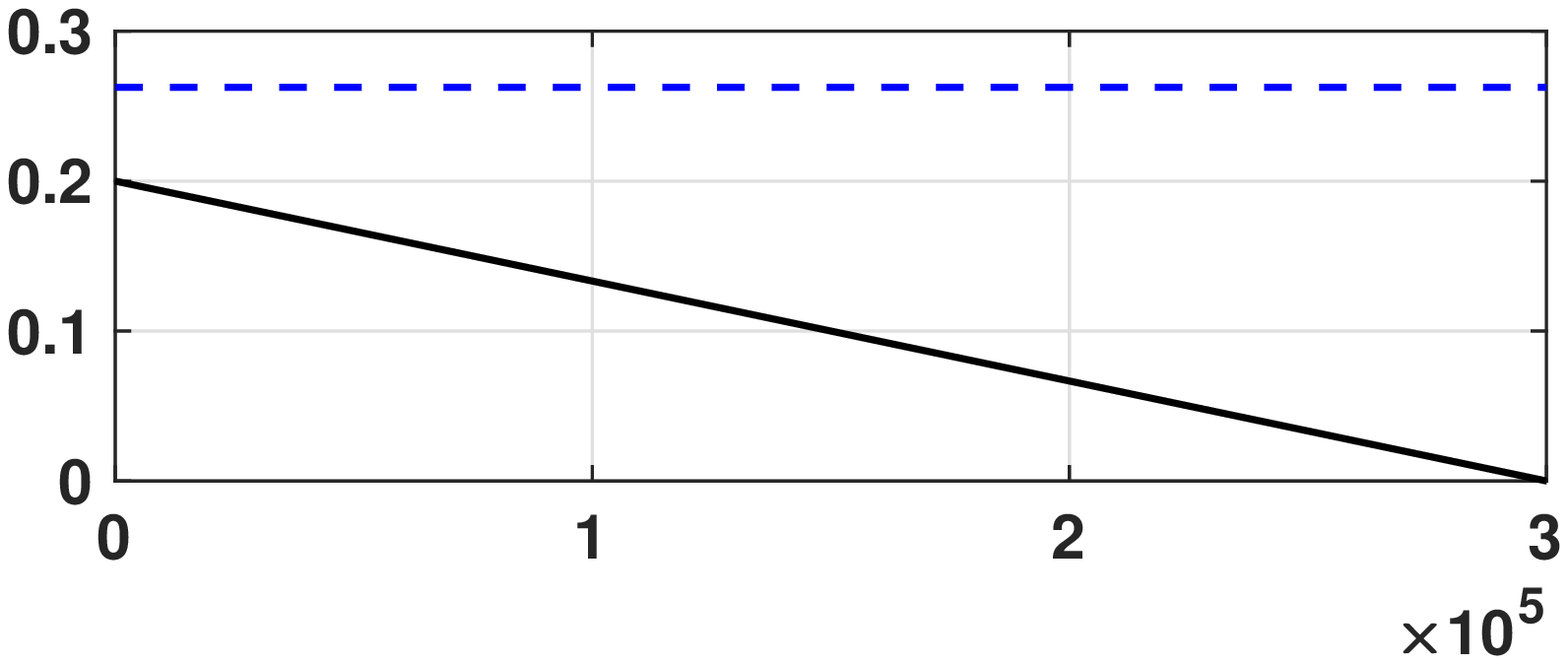}}\hfill
\subfloat[Existence of a positive local stable point $N_v^*$.  \label{fig_g1}]{\includegraphics[width=0.49\textwidth]{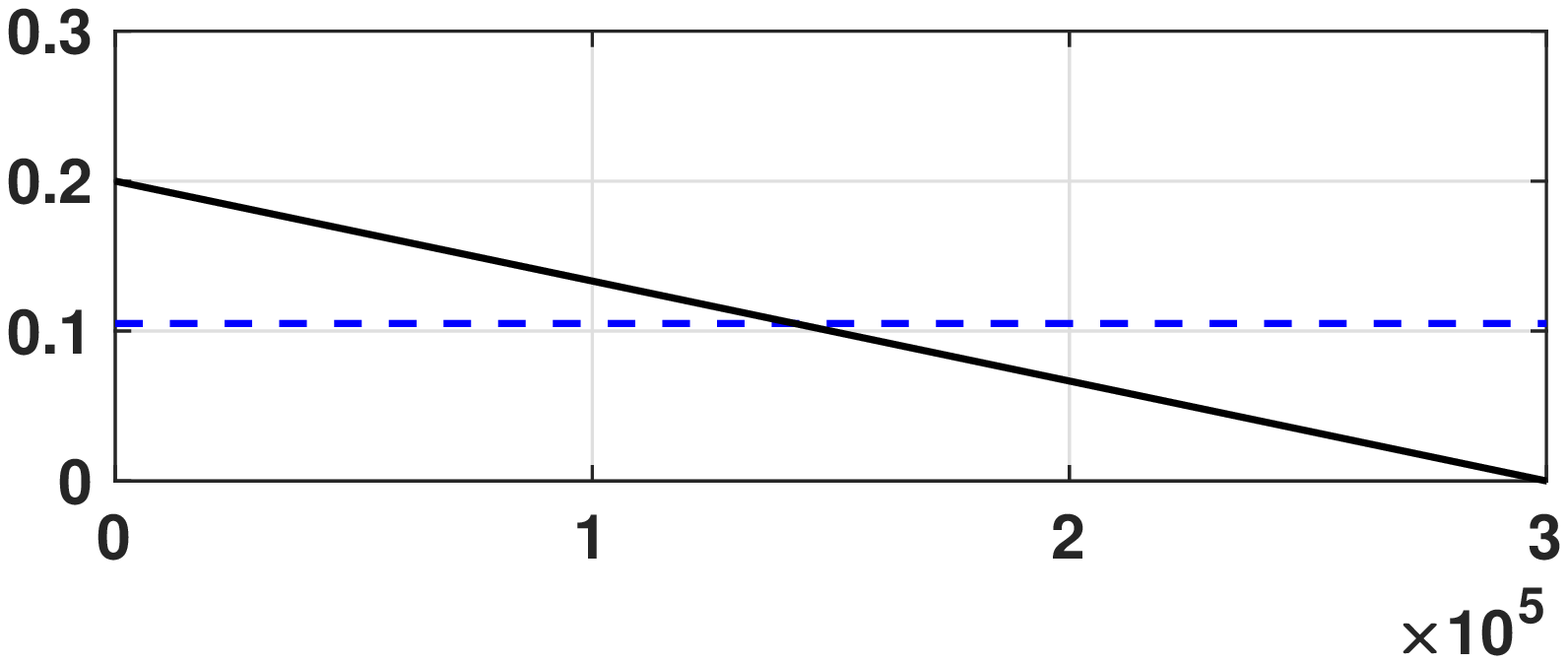}}\hfill
\caption{$g(N_v)$ (solid line) as a function of the number of vectors $N_v$, considered in Section \ref{sec:numexp}. The dashed line corresponds to the value $\phi$. Solutions to the equation $g(N)=\phi$ correspond to non-trivial steady states for $N_v$.
\label{fig_g}}
\end{figure}

\begin{example} \label{ex1} {\rm 
We first confirm that the condition $\mathcal{R}_0>1$ could lead to an endemic state, as long as there is a stable positive steady-state for vectors. We consider a set of parameters for which $\phi \approx 3.39$ and $\mathcal{R}_0 \approx 1.60$. First, if $r=0.20$, the only stable fixed point is $(E^*,S_v^*,I_v^*)=(0,0,0)$ for which $i^*(a)=0$; see Figures \ref{fig:ex1a}, \ref{fig:ex1b}. Second, if $r=5$, we have the stable fixed point $(E^*,S_v^*,I_v^*) \approx (9633,7505,88824)$. In this case, we have an endemic state as shown in Figures \ref{fig:ex1c}, \ref{fig:ex1d}, according to Theorem \ref{th:endemic}.

\begin{figure}[!ht]
\centering
\subfloat[ $i_h(t,a)$ \label{fig:ex1a}]{\includegraphics[width=0.49\textwidth]{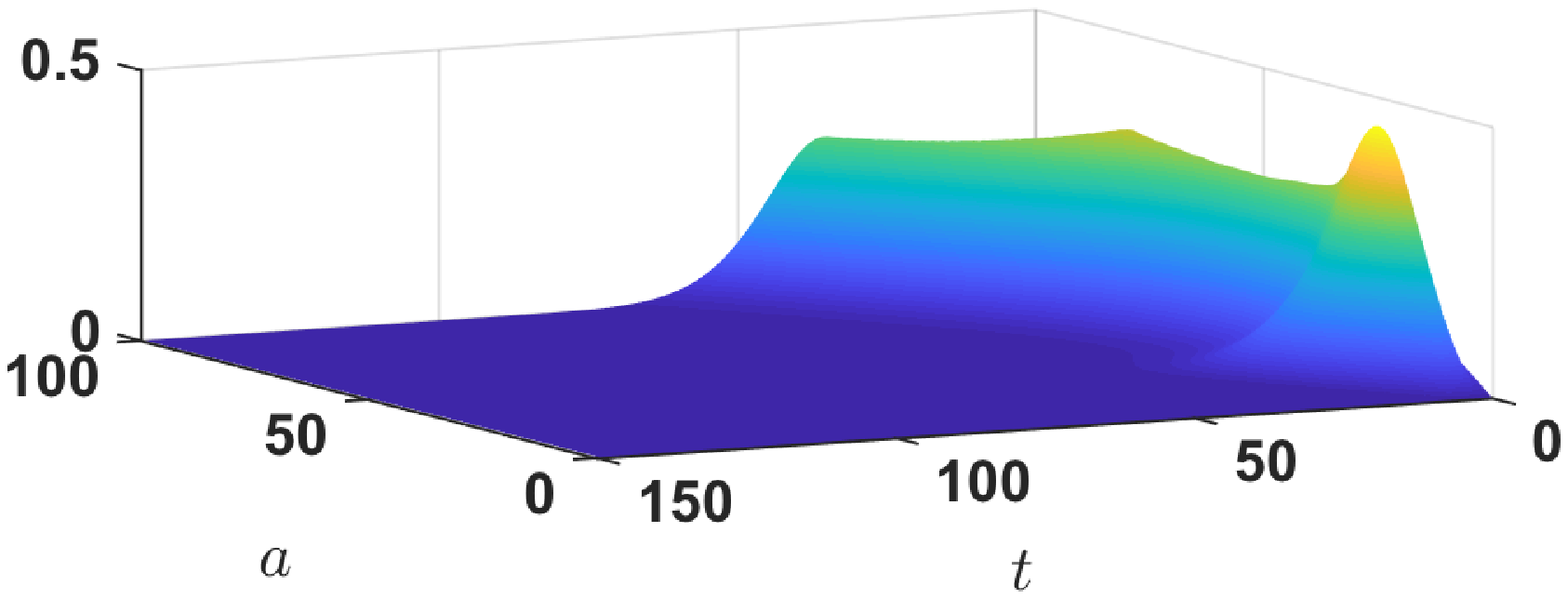}}\hfill
\subfloat[Vectors\label{fig:ex1b}]{\includegraphics[width=0.49\textwidth]{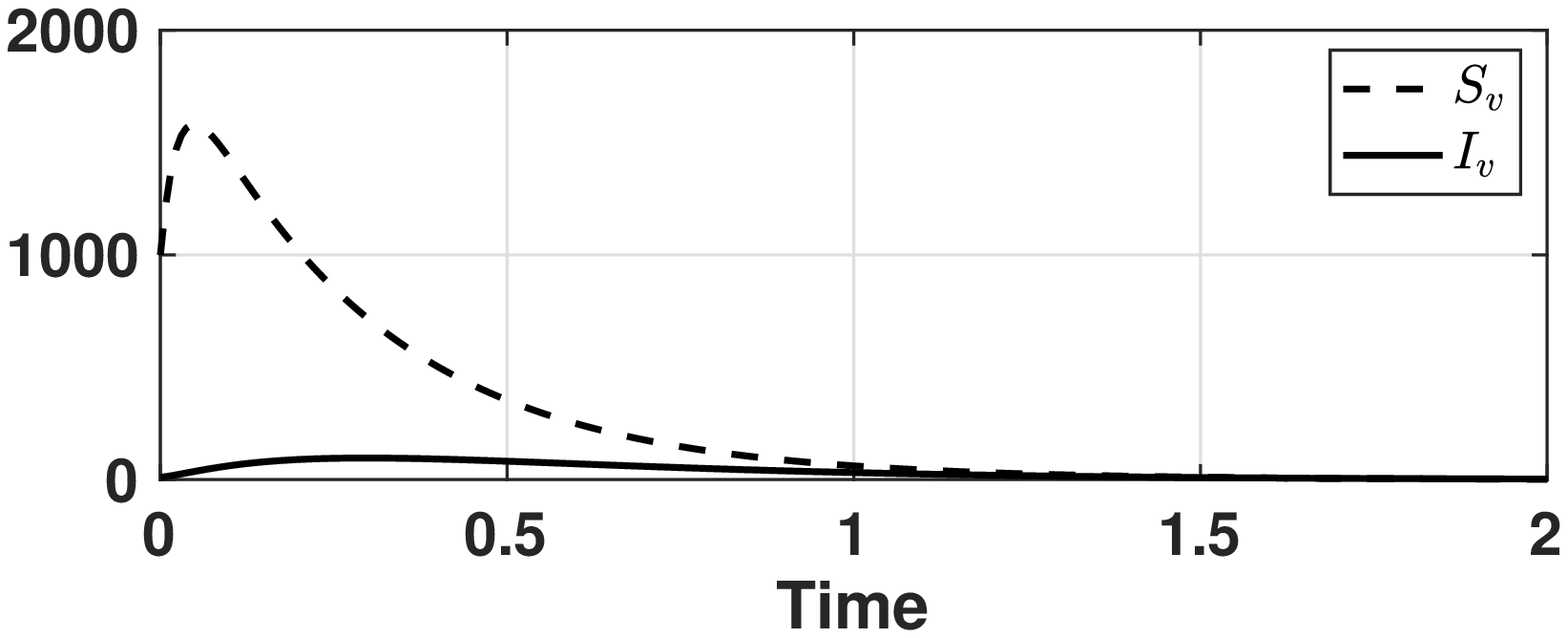}}\hfill
\subfloat[ $i_h(t,a)$ \label{fig:ex1c}]{\includegraphics[width=0.49\textwidth]{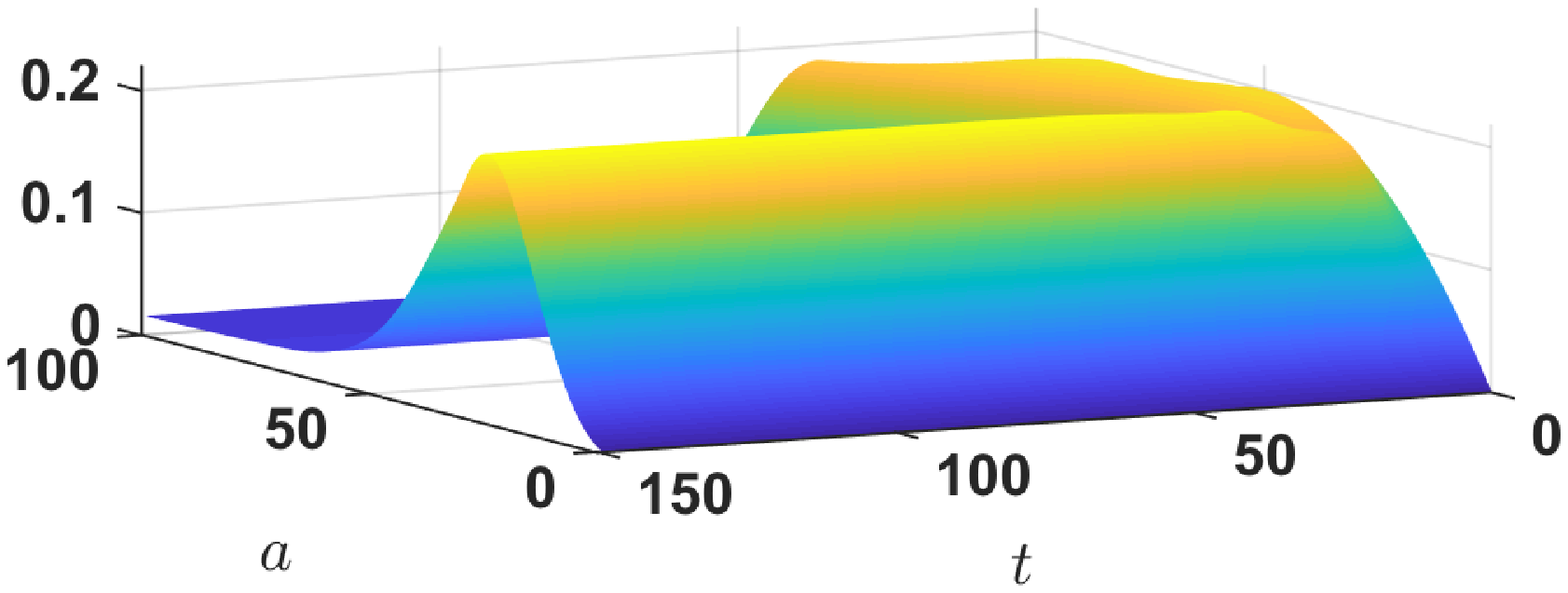}}\hfill
\subfloat[Vectors\label{fig:ex1d}]{\includegraphics[width=0.49\textwidth]{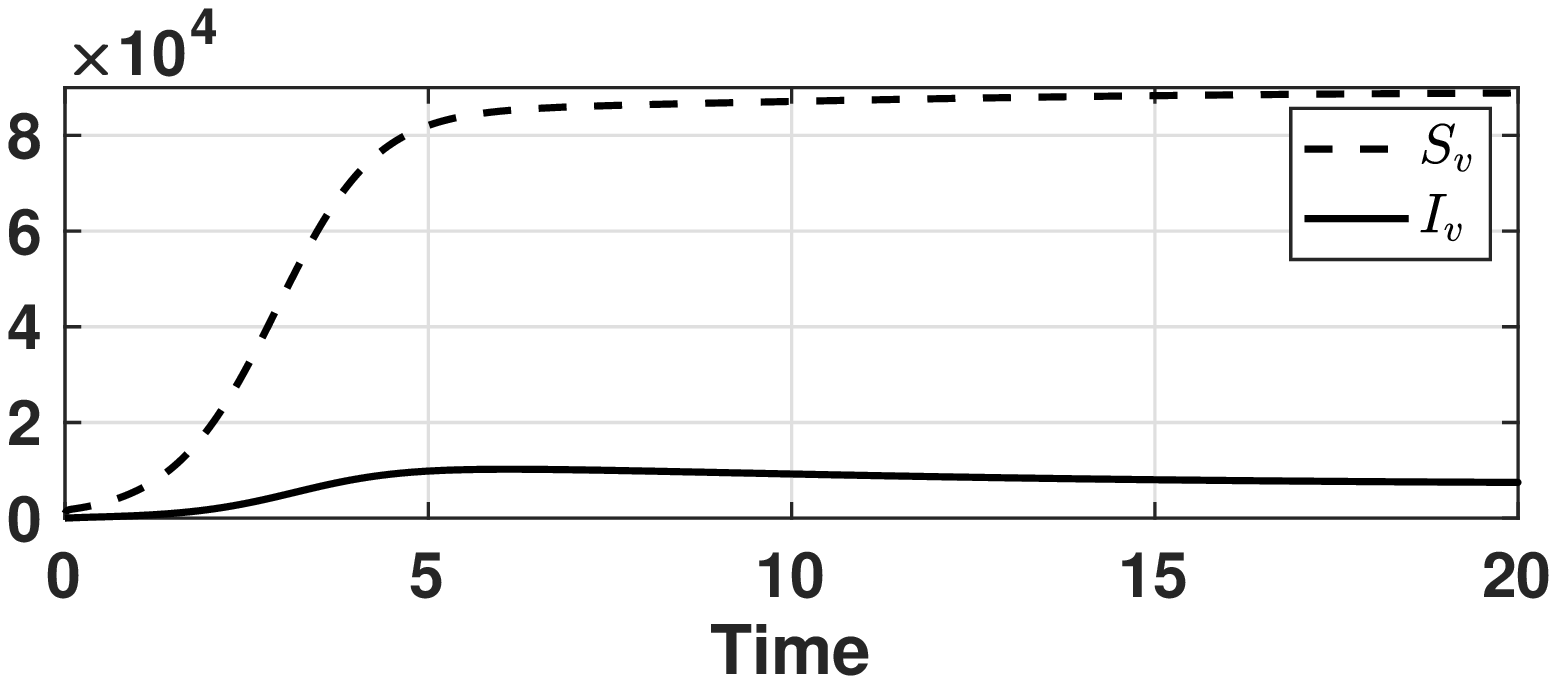}}\hfill
\caption{Solutions for the (left) infected class and (right) vectors when $\mathcal{R}_0>1$ for (top) $\phi>r=0.20$ and (bottom) $\phi < r=5$; see Example \ref{ex1}. \label{fig:ex1}}
\end{figure}
}
\end{example}

\begin{example} \label{ex2}
{\rm
In this example we confirm that the condition $\mathcal{R}_0<1$ is sufficient to guarantee a disease-free steady state. We take $r=5$ and reduce $\beta$ such that $\mathcal{R}_0<1$. The infected class reaches a disease-free state as shown in Figure \ref{fig:ex2a}, according to Theorem \ref{th:R0}. Even though there exists a positive state for vectors $(E^*,N_v^*)$ as shown in Figure \ref{fig:ex2b}, we observe that $I_v^*=0$.

\begin{figure}[!htb]
\centering
\subfloat[$i(t,a)$ \label{fig:ex2a}]{\includegraphics[width=0.44\textwidth]{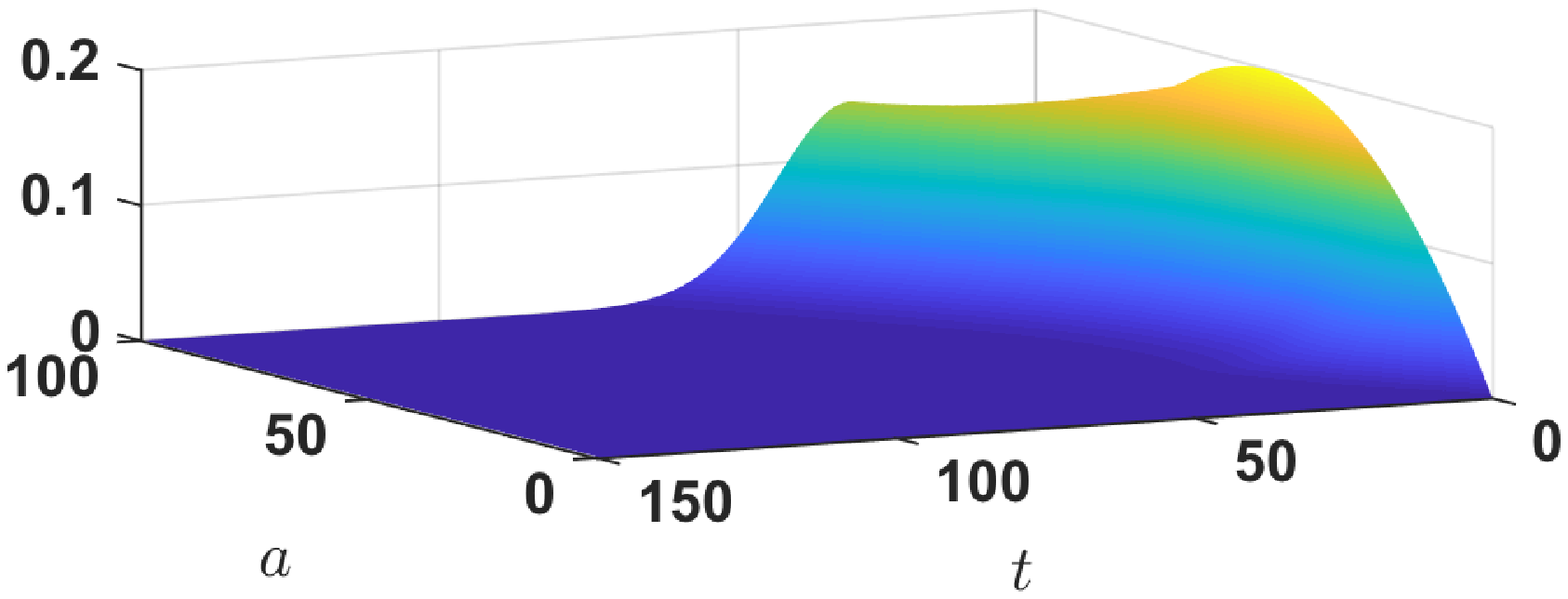}}\hfill
\subfloat[Vectors \label{fig:ex2b}]{\includegraphics[width=0.44\textwidth]{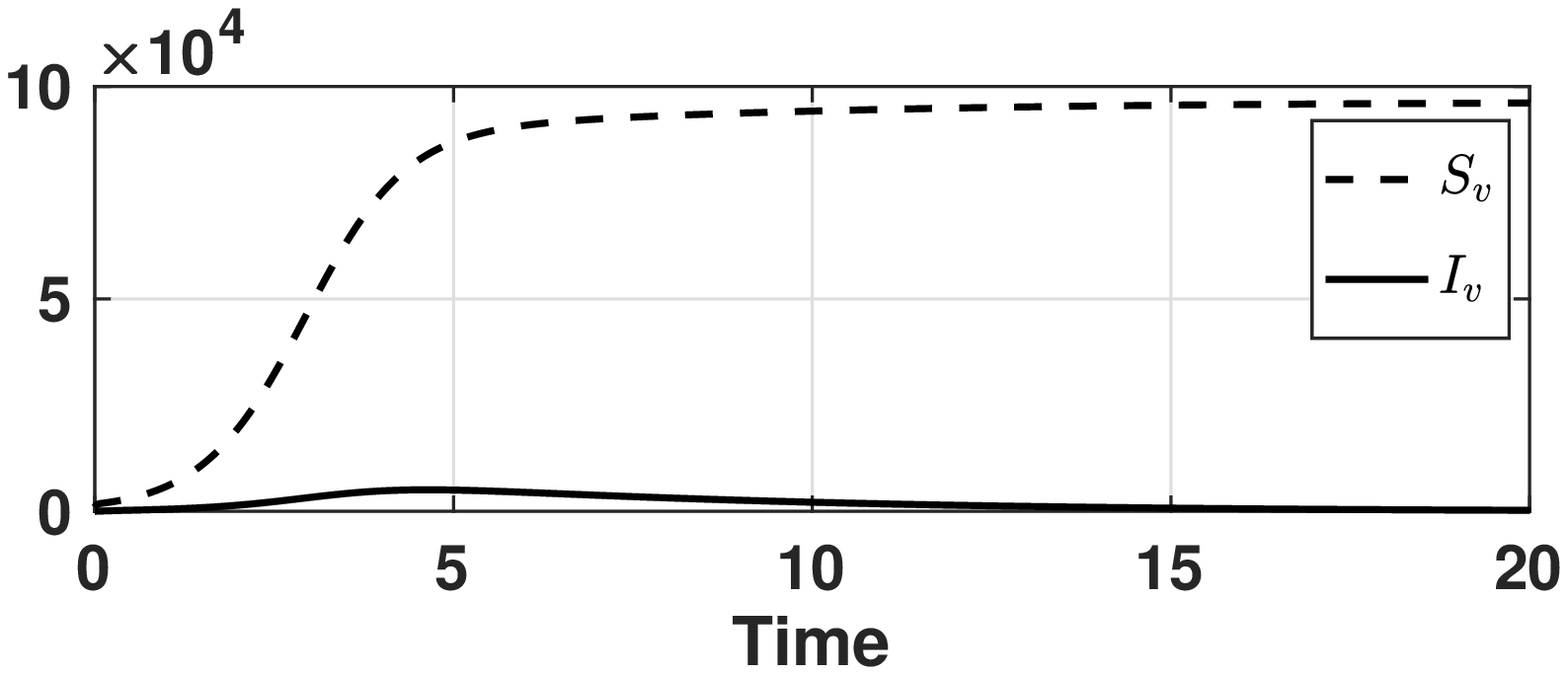}}\hfill
\caption{When $\mathcal{R}_0<1$, the disease-free state is stable, even though there is a positive steady state for vectors ($\mathcal{R}_v>1$); see Example \ref{ex2}. \label{fig:ex2}}
\end{figure} 
}
\end{example}

\begin{example} \label{ex3} {\rm
We now consider the case $\mathcal{R}_0>1$ with initial conditions $E_0=0$, $S_{v_0} = 10$, $I_{v_0} = 1$, $i_h(0,a) = r_h(0,a) =0$ (no infected or immune humans at time $t=0$); see results for the infected class in Figure \ref{fig:ex3}. It is clear that $\mathcal{R}_0>1$ guarantees an endemic state as long as vectors can survive.

\begin{figure}[!h]
\centering
\subfloat[$i_h(t,a)$ \label{fig:ex3a}]{\includegraphics[width=0.49\textwidth]{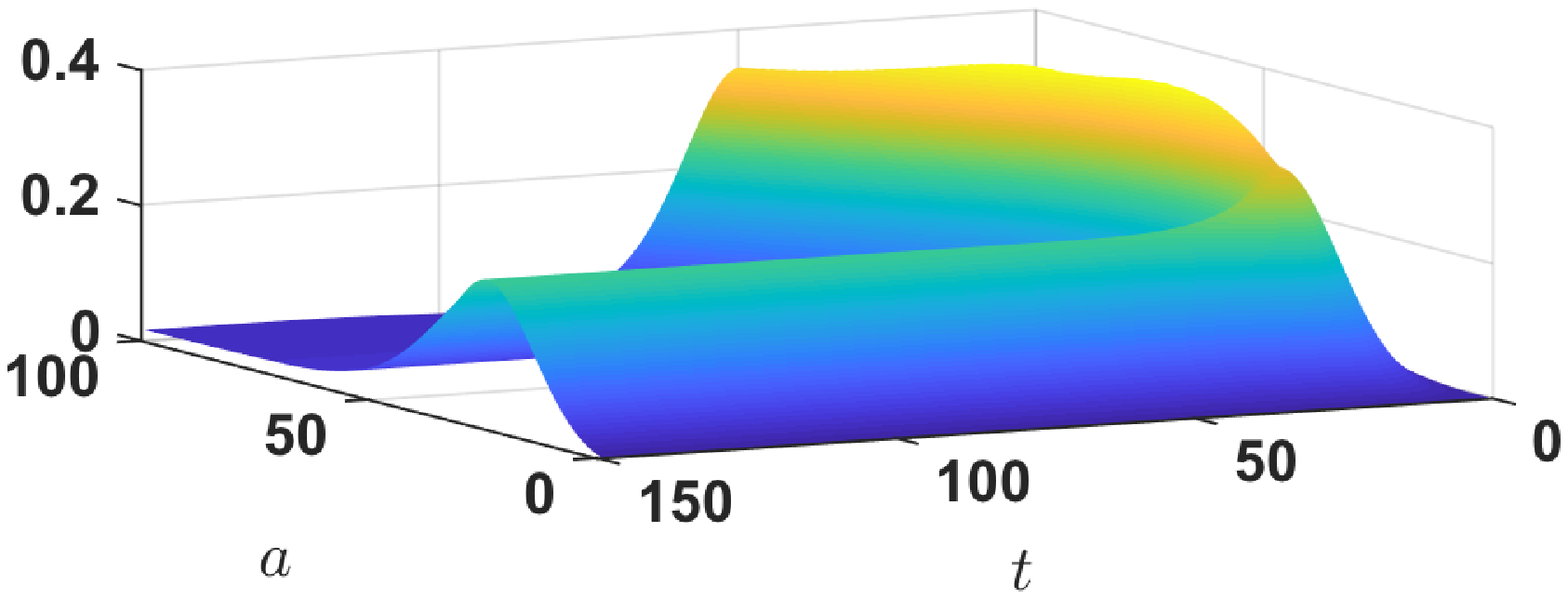}}\hfill
\subfloat[Vectors\label{fig:ex3b}]{\includegraphics[width=0.49\textwidth]{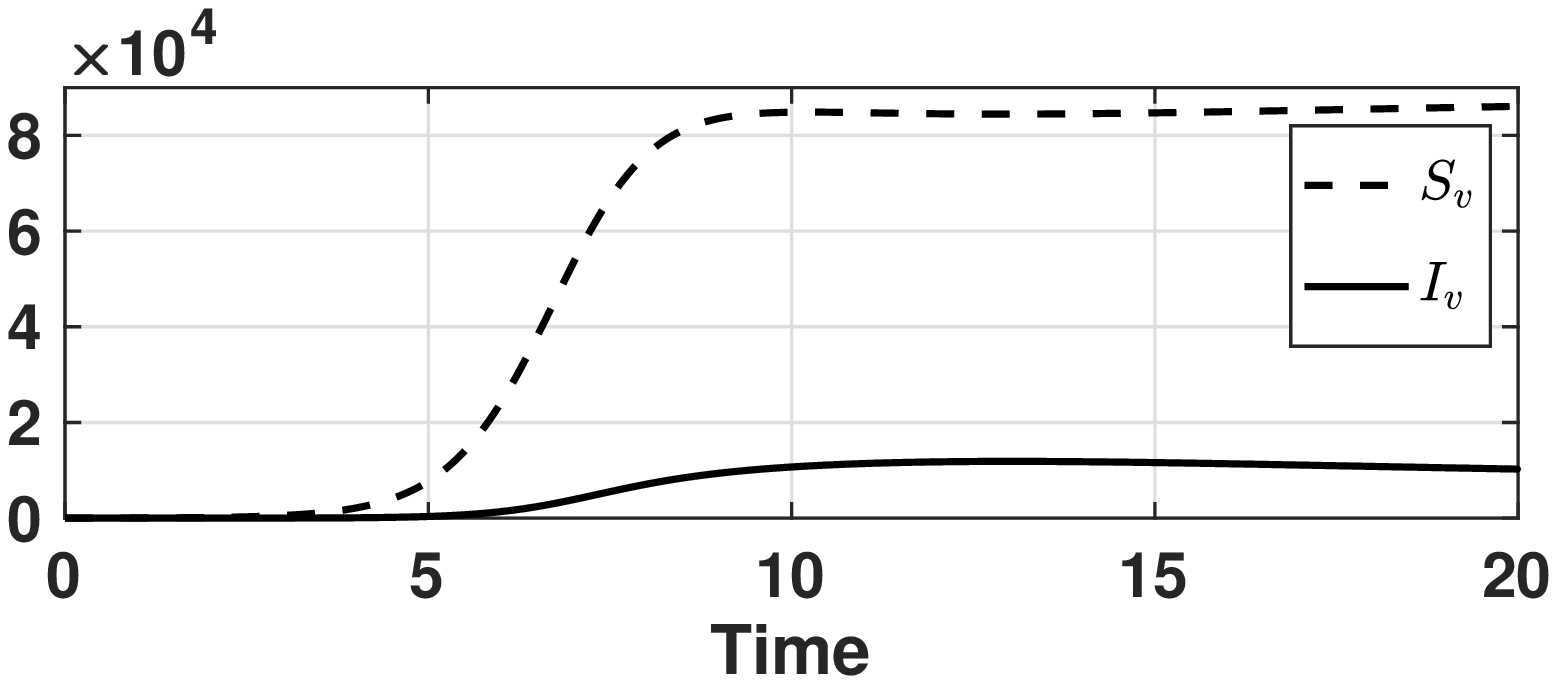}}\hfill
\caption{When $\mathcal{R}_0>1$, if there exists a positive state for vectors ($\mathcal{R}_v>1$) we can observe an endemic state on humans; see Example \ref{ex3}. \label{fig:ex3}}
\end{figure} 

}
\end{example}

\subsection{Multiple vector demographic states}
In a second set of experiments we use 
\begin{equation} \label{eq:g2}
g(N_v) = r e^{-N_v/c_1} (\sin(c_2 N_v)+1),
\end{equation}
for given constants $r$, $c_1$, $c_2$. Here, $r$ is the vector per-capita fertility rate, $c_1$ is a form of vector control and $c_2$ represents the variations in vector densities; for a particular choice of parameters see Figure \ref{fig:stabNv*}. Equation \eqref{eq:g2} represents the different growth rates of vectors for the wet and dry seasons. In this way, we simulate variations based on vector control efforts, obtaining multiple vector demographic steady states for $(E^*,N_v^*)$. For the particular choice of parameters we have used, we obtain eight positive fixed points. Numerically we confirm that four of them are locally stable.



\begin{example} {\rm \label{ex4}
We first confirm the result proved in Lemma \ref{lem:Rv0}. We observe that if $r<\phi$, the infection-free steady state is stable and unstable otherwise; results for $(E_0,S_{v0},I_{v0})=(10,10,10)$ are shown in Figure \ref{fig:ex4a}. We then obtain different solutions for different initial conditions for the vector classes for which different positive steady-states are reached; see Figure \ref{fig:g2_difI}. Despite having multiple vector densities the outbreaks are similar in severity. This implies that even when vector density is low an outbreak is possible.

%

\begin{figure}[ht!]
\centering
\subfloat[Vector growth rate ($r=3$) is less than $\phi$. \label{fig:g2_difIaa}]{\includegraphics[width=0.45\textwidth]{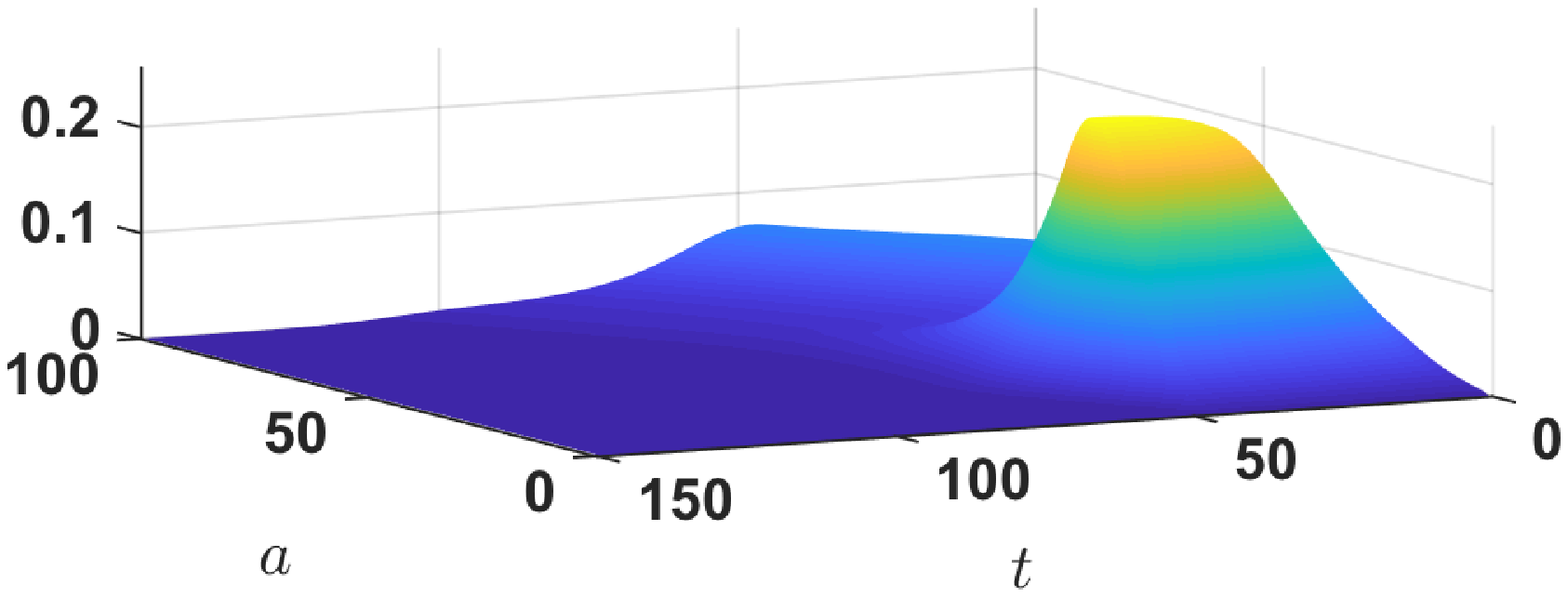}}\hfill
\subfloat[Vector growth rate ($r=4$) is greater than $\phi$. \label{fig:g2_difIb}]{\includegraphics[width=0.45\textwidth]{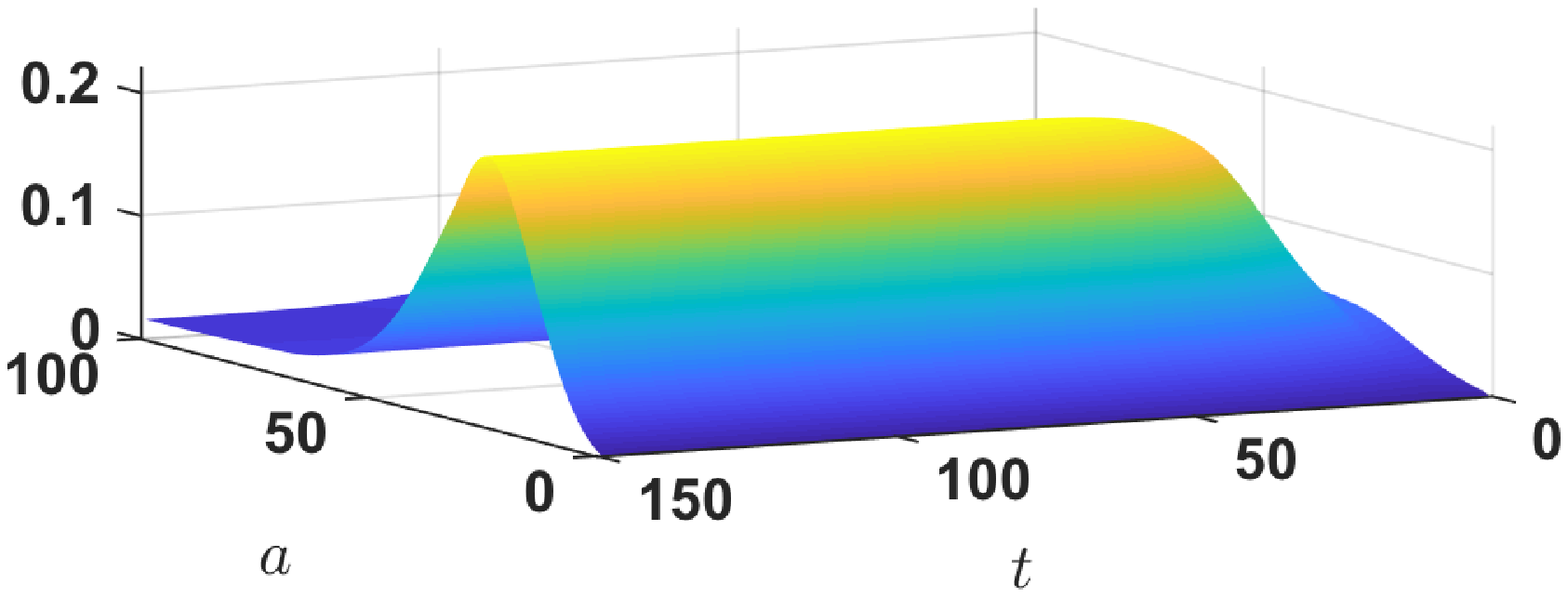}}
\caption{Different steady state distributions of infected individuals $i_h^*(t,a)$ with initial conditions $(E_0,N_{v0}) = (10,20)$. If $r<\phi$, the infection-free steady state is stable and unstable otherwise; see Example \ref{ex4}. \label{fig:ex4a}}
\end{figure}

\begin{figure}[h!]
\centering
\subfloat[$(E_0,N_{v0}) = (1525,27460)$ \label{fig:g2_difIc}]{\includegraphics[width=0.49\textwidth]{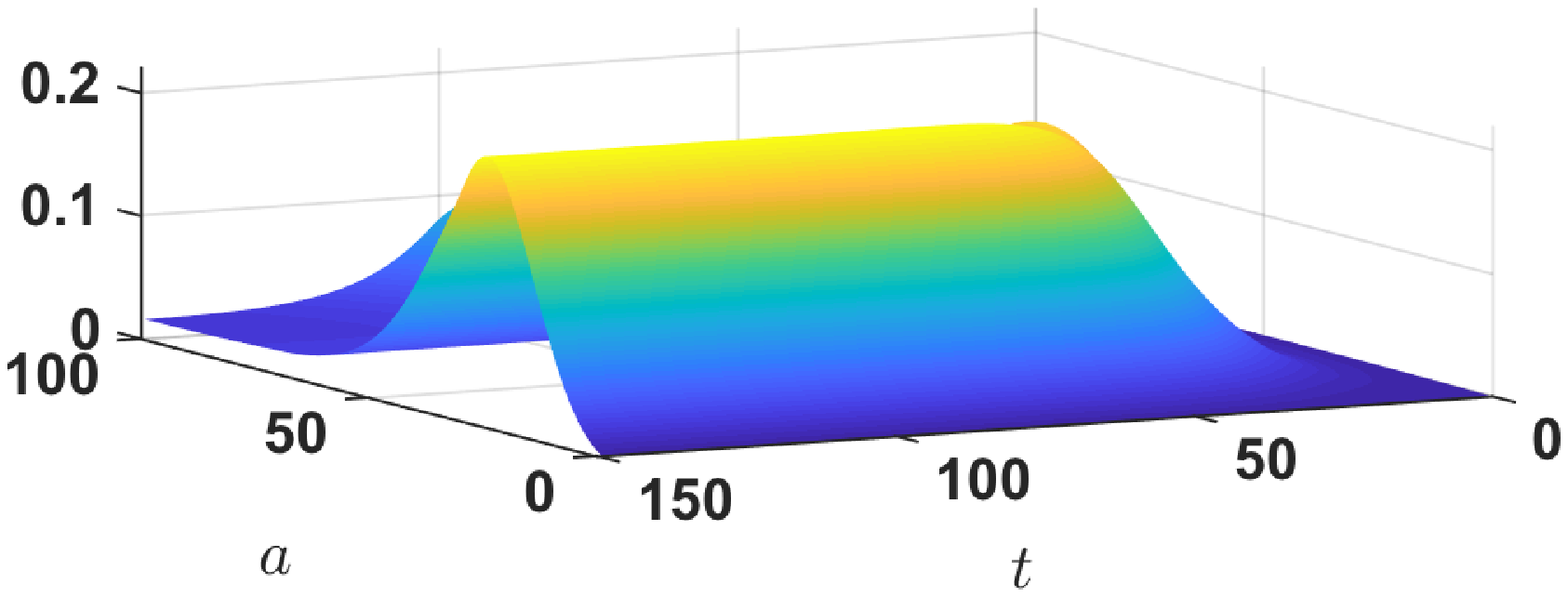}}\hfill
\subfloat[Vectors \label{fig:zeroUns}]{\includegraphics[width=0.49\textwidth]{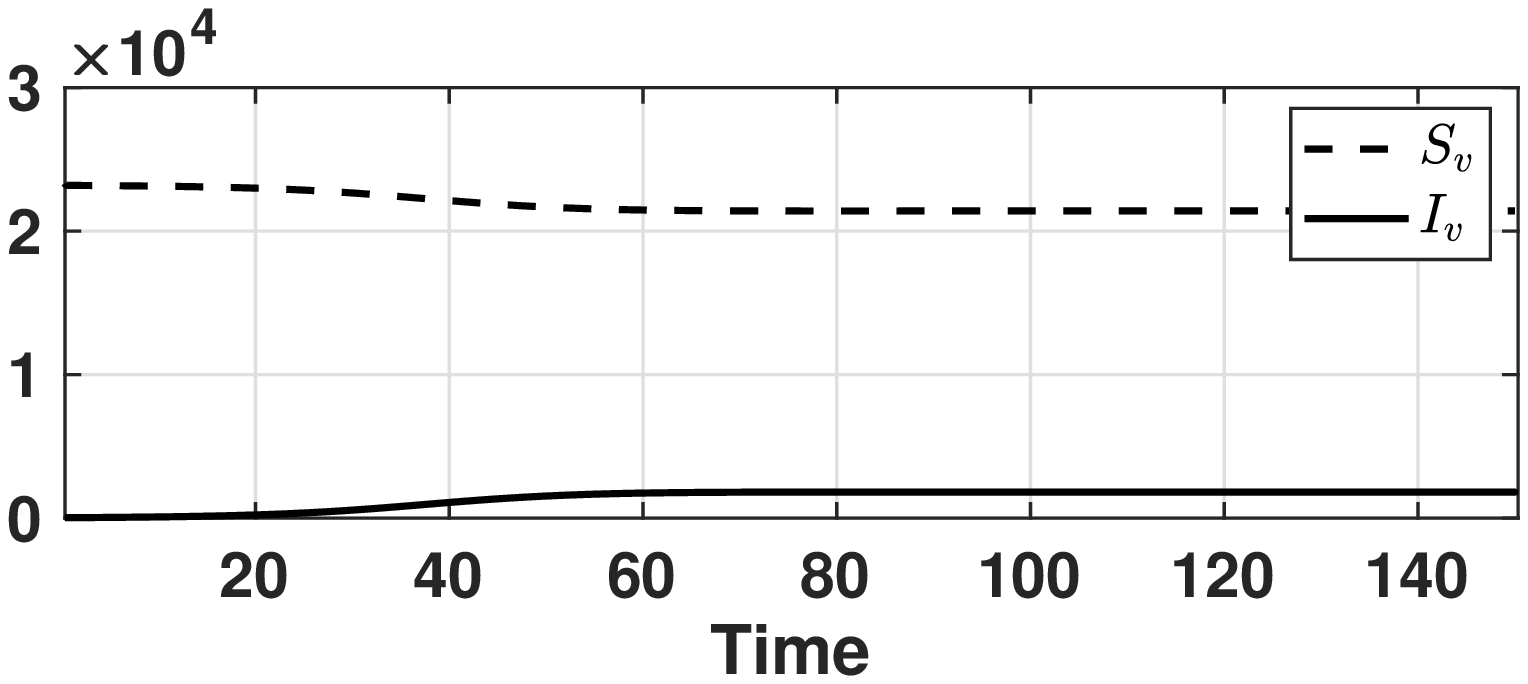}}\hfill
\subfloat[$(E_0,N_{v0}) = (200000,100000)$ \label{fig:g2_difIa}]{\includegraphics[width=0.49\textwidth]{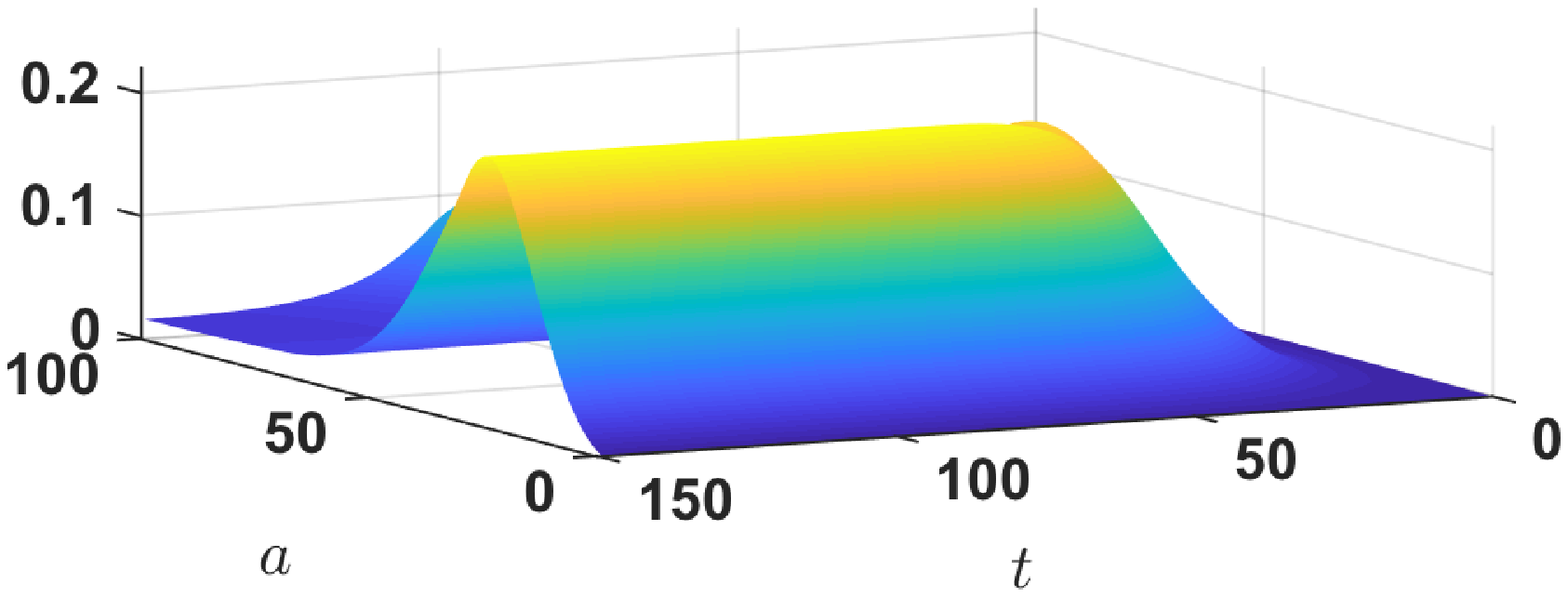}}\hfill
\subfloat[ \label{fig:g2_difId}]{\includegraphics[width=0.49\textwidth]{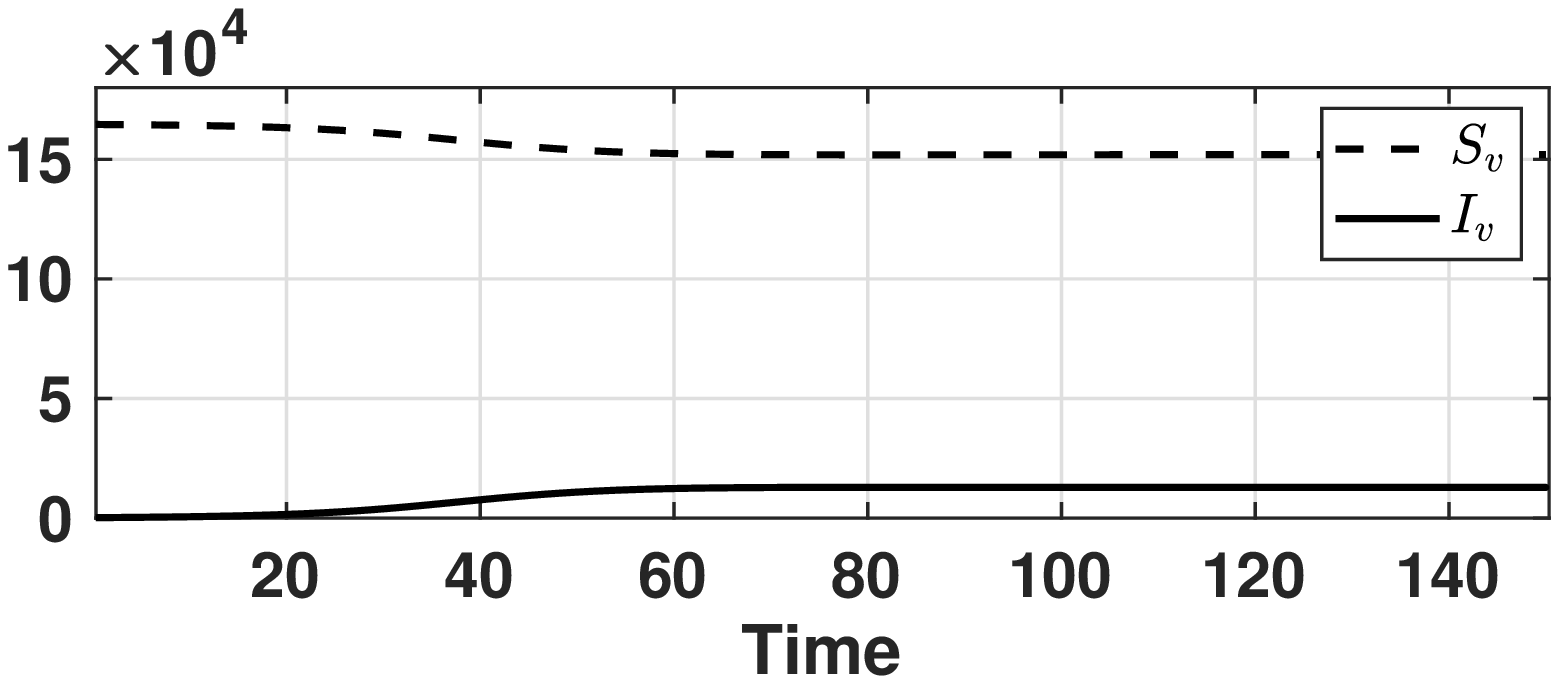}}
\caption{(left) Infected individuals $i_h(t,a)$ and (right) vector solutions with different initial conditions. For this choice of parameters, $\mathcal{R}_0>1$, $\mathcal{R}_v>1$ and several vector steady states exist; see Example \ref{ex4}. 
\label{fig:g2_difI}}
\end{figure}
}
\end{example}

\begin{example} \label{ex5} {\rm
Similarly as Example \ref{ex2}, we confirm that $\mathcal{R}_0<1$ is sufficient for the disease to die out, even in the presence of a positive population of vectors. In this case, $(E^*,N_v^*) = (7430, 74305)$, but $I_v^*=0$; see Figure \ref{fig:ex5}.  Even when the {\it demographic vector number} is bigger than one the disease can be under control when $\mathcal{R}_0<1$. 

\begin{figure}[h!]
\centering
\subfloat{\includegraphics[width=0.49\textwidth]{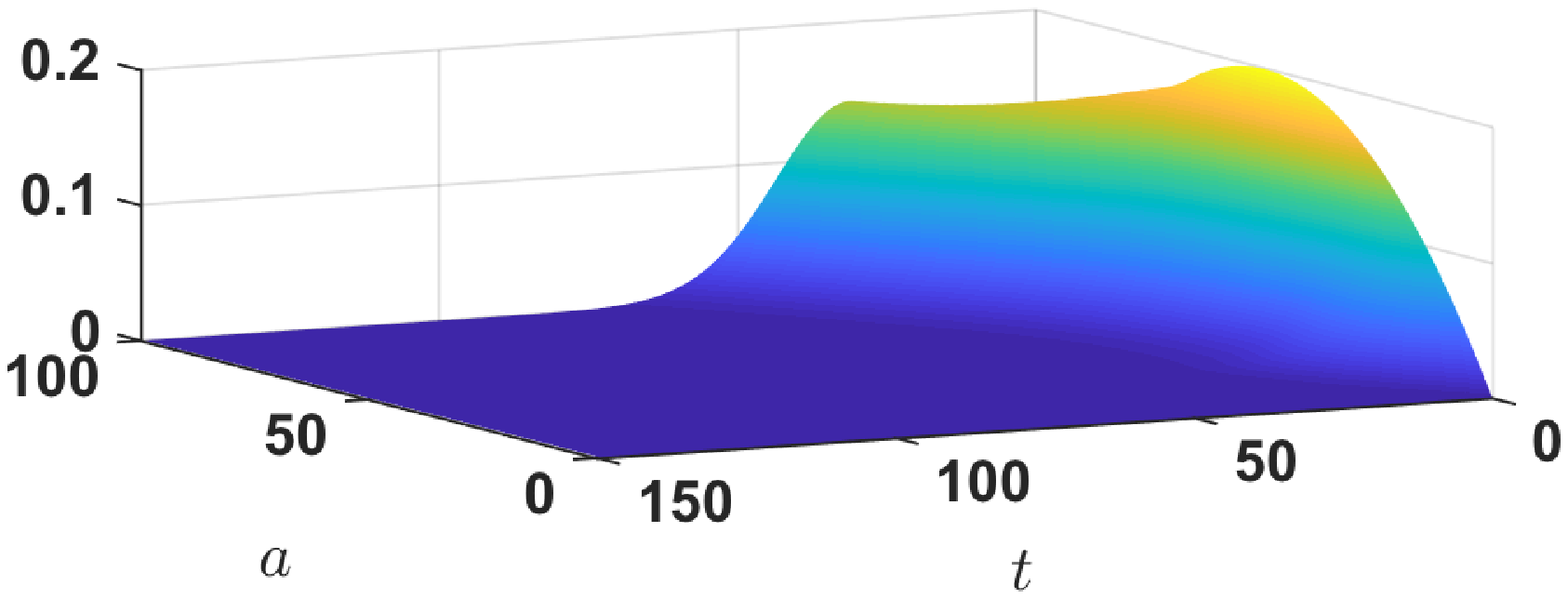}}\hfill
\subfloat{\includegraphics[width=0.49\textwidth]{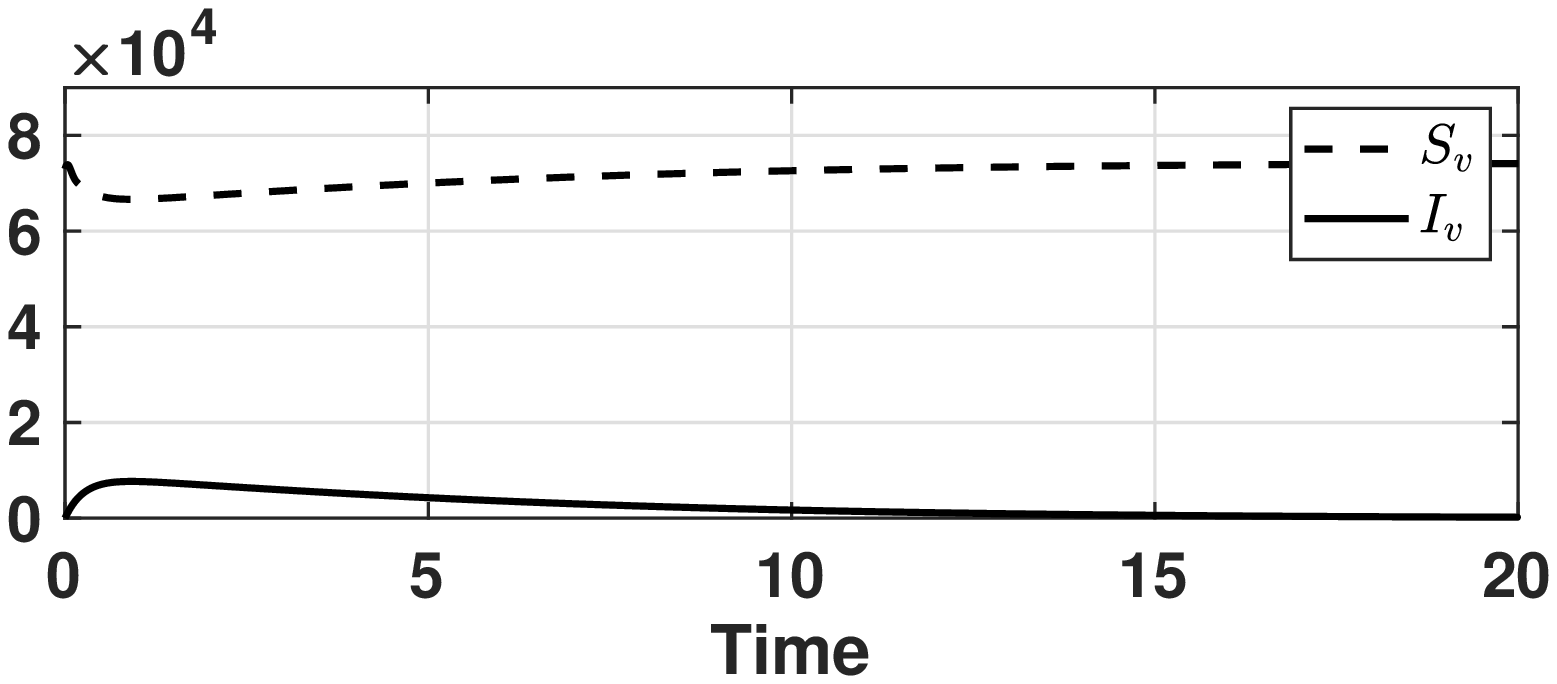}}\hfill
\caption{$\mathcal{R}_0<1$ is sufficient to guarantee that (left) $i_h(a)=0$ and (right) $I_v^*=0$, even though $S_v^*>0$, $E^*>0$; see Example \ref{ex5}. \label{fig:ex5}}
\end{figure} 
}
\end{example}

\begin{example} \label{ex6} {\rm
Similarly as Example \ref{ex3}, we confirm that $\mathcal{R}_0>1$ implies the existence of an endemic state, as long as a positive equilibrium state exists for the vectors to survive; see Figure \ref{fig:ex6}.

\begin{figure}[h!]
\centering
\subfloat{\includegraphics[width=0.49\textwidth]{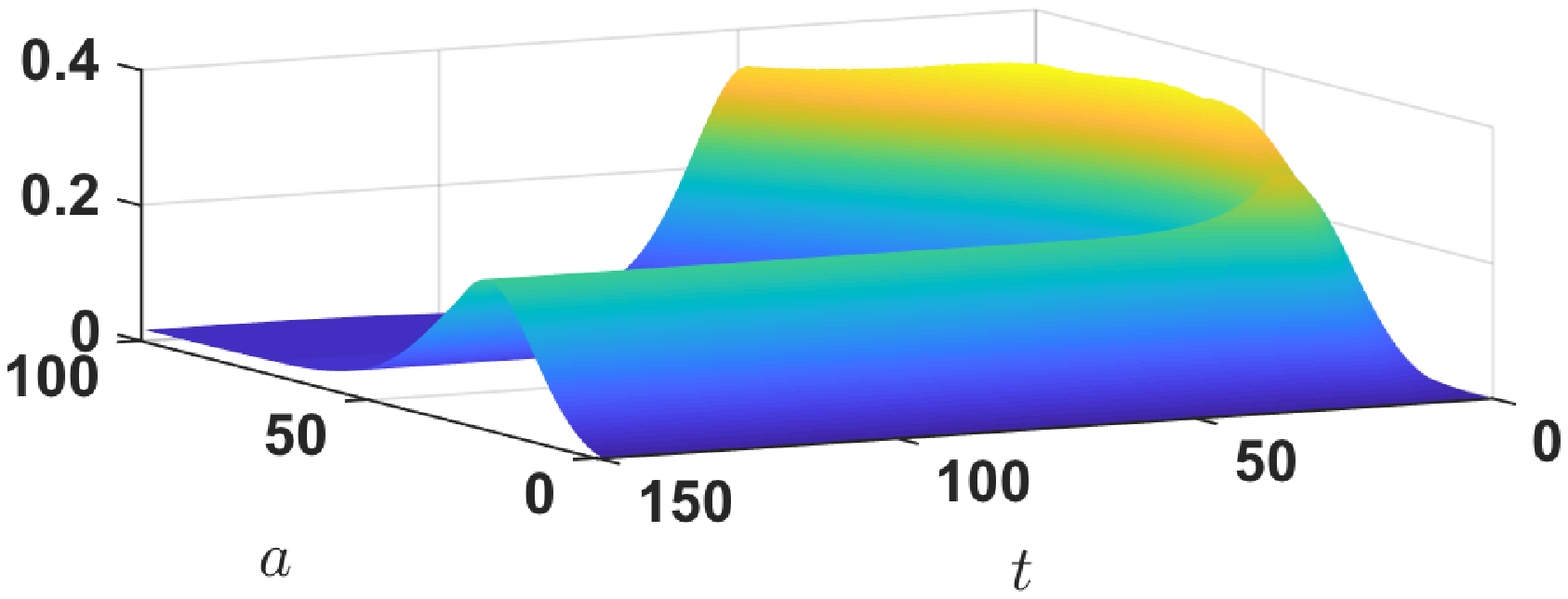}}\hfill
\subfloat{\includegraphics[width=0.49\textwidth]{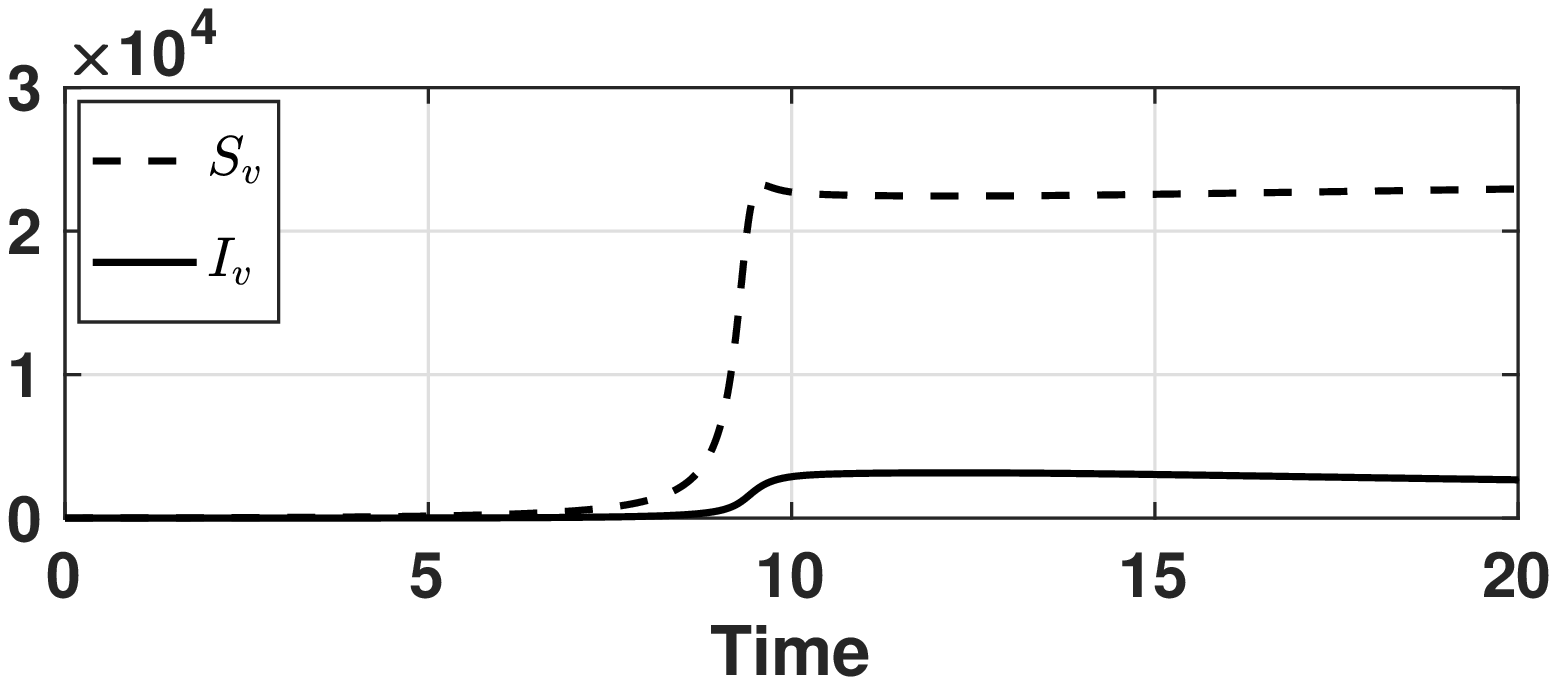}}\hfill
\caption{(Left) $i_h(t,a)$ and (right) $(E(t),S_v(t),I_v(t))$ for $\mathcal{R}_0>1$; see Example \ref{ex6}. Initially all humans are susceptible and $(E_0,S_{v0},I_{v0})=(0,10,1)$. \label{fig:ex6}}
\end{figure} 

}
\end{example}

\subsection{Effect of seasonality on dengue dynamics} \label{sec:seasons}
In most places where dengue is endemic, seasonal variations in vector populations play a major role in disease transmission. Moreover, it determines the distribution of resources allocated for preventive/control measures. Typically, dengue incidence is correlated with the rainy season. The importance of understanding seasonal variations per location could potentially help public health officials to allocate resources, as well as having better preventive/control measures to reduce dengue incidence (focused primarily towards the reduction of vector breeding sites).

\begin{figure}[ht!]
\centering
\subfloat{\includegraphics[width=0.49\textwidth]{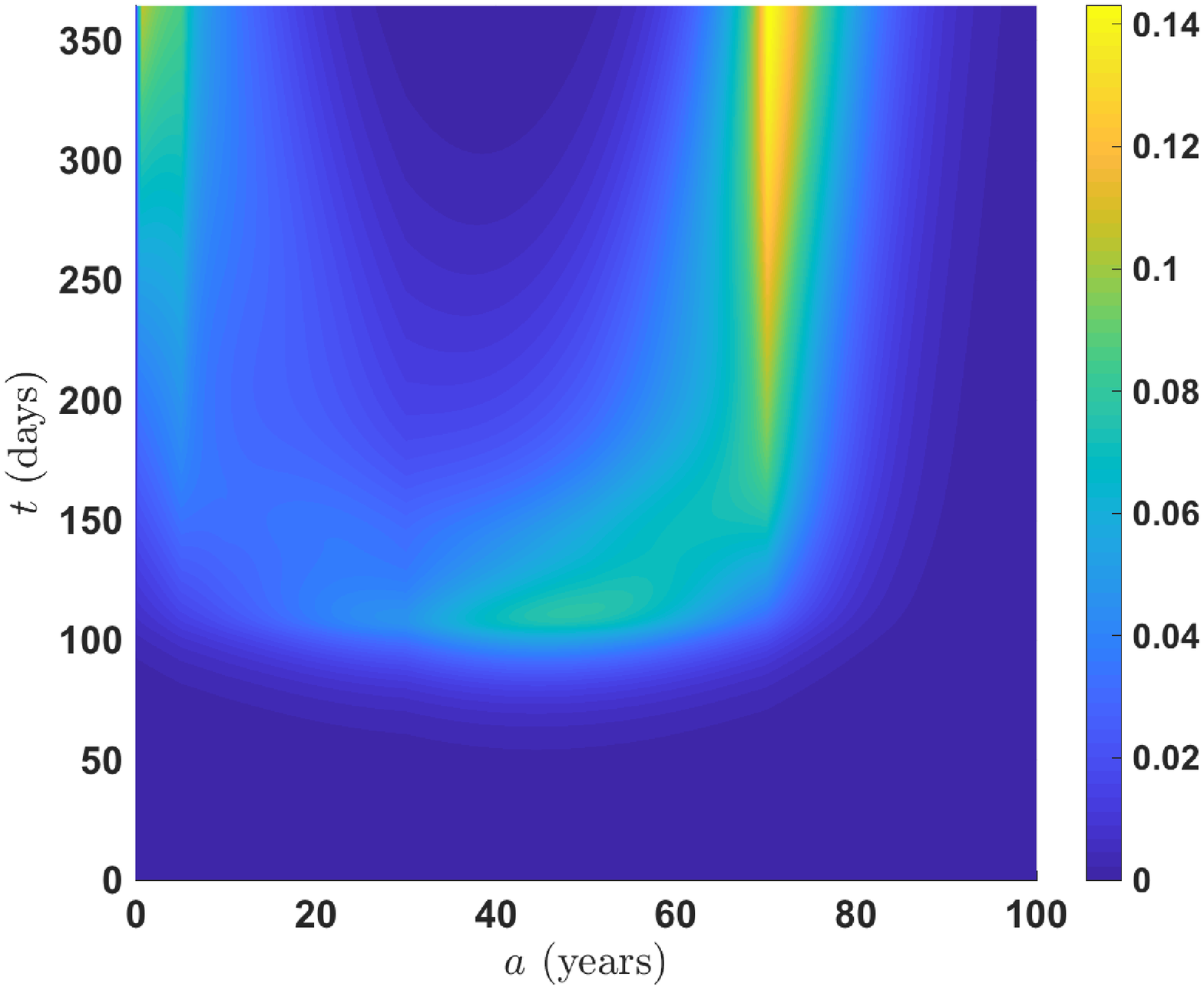}}\hfill
\subfloat{\includegraphics[width=0.49\textwidth]{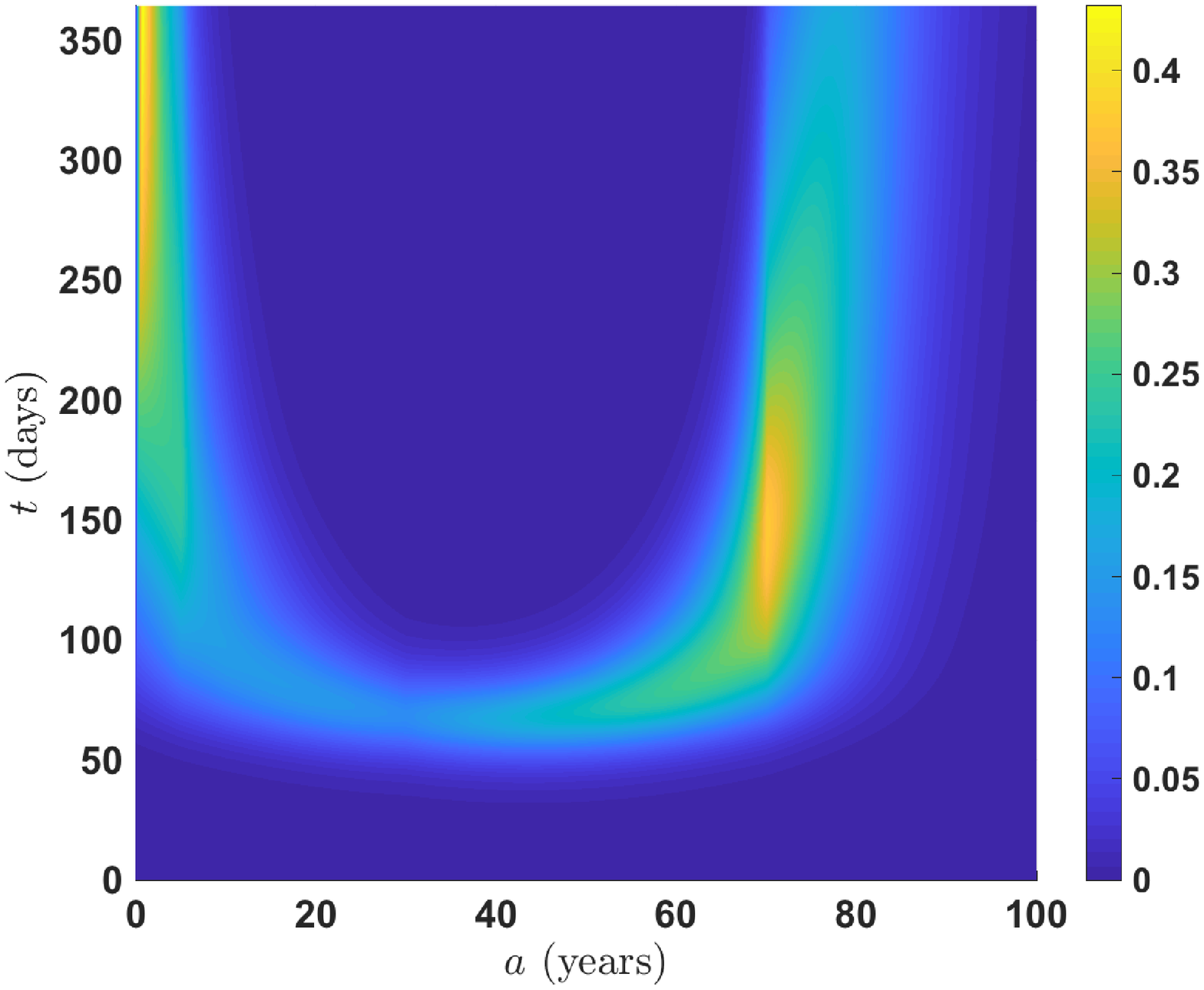}}\hfill
\caption{Solution for $i(t,a)$ when seasonal parameters are considered, for (left) $\beta_0 = 0.09$ and (right) $\beta_0 = 0.19$; see Section \ref{sec:seasons}. \label{fig:ex8}}
\end{figure} 

We include some numerical results where $g(N_v)$, $\beta_v$ and $\delta$ depend periodically on time, simulating high and low seasons in the dynamics of vectors. Here we use parameters for the vector classes as in \cite{sanchez2006}. We consider a population with only susceptible humans. In the vector classes, we include one infected vector in order to observe the propagation of the disease; see results in Figure \ref{fig:ex8}. For different values of the transmission rate ($\beta$(a)) the infected host distribution distinctly affects the younger and senior age groups.

\section{Discussion} \label{sec:disc}
We have constructed a model with age-structure within host and early-life stage of the vector with a general function $f(N_v)$ that represents the new vectors from the egg/larvae stage in the vector system. This gives the possibility of multiple demographic steady states for the vector population and its stability depends on the {\it vector demographic number}, $\mathcal{R}_v$. The local stability of the vector-free state when $R_v<1$ was established. The {\it basic reproductive number} was computed and the local and global asymptotic stability of the disease-free equilibrium was determined when $\mathcal{R}_0<1$.

When $\mathcal{R}_0>1$ and we have a stable vector demographic steady state ($\mathcal{R}_v>1$), the disease is then endemic. Control measures on the early-life stage of the vector can guarantee an adult vector-free state and hence, the disease dies out. Vector control measures such that $\phi > \max g(N_v)$ implies that vectors will die out independently on the value of $\mathcal{R}_0$.

There are important public health implications when we are able to include host age distribution, which can determine better strategies for hospitalized individuals. Furthermore, control measures on the early-life stage of the vector can effectively change the landscape on how public health officials lead prevention efforts before the onset of a dengue outbreak.

\section{Acknowledgements}
We thank the Research Center in Pure and Applied Mathematics and the Mathematics Department at Universidad de Costa Rica for their support during the preparation of this manuscript. The authors gratefully acknowledge institutional support for project B8747 from an UCREA grant from the Vice Rectory for Research at Universidad de Costa Rica. 



\end{document}